\documentclass[journal,twoside,web]{ieeecolor}
\usepackage{generic}
\usepackage{amsmath}
\usepackage{graphicx}
\usepackage{algorithmic, bm}
\usepackage{array}
\usepackage{interval}
\usepackage{amssymb,amsfonts,mathrsfs,mathtools}
\usepackage{subcaption}
\usepackage{hyperref}
\usepackage{color}
\usepackage{booktabs,makecell}
\usepackage{pict2e}
\usepackage{optidef}
\usepackage{algorithm}
\usepackage{doi}
\usepackage{flushend}
\hypersetup{colorlinks = true,
 linkcolor = red,
 urlcolor = black,
 citecolor = blue,
 anchorcolor = blue}
 \usepackage{textcomp}
 \usepackage{cite}
\def\BibTeX{{\rm B\kern-.05em{\sc i\kern-.025em b}\kern-.08em
 T\kern-.1667em\lower.7ex\hbox{E}\kern-.125emX}}
\newcommand\jw{j\omega}

\newcommand\nn{{n \times n}}
\newcommand\semicir{\mathbf{SC}(\epsilon, \jw_0)}
\newcommand\bbkt[1]{\left\{#1\right\}}
\newcommand\sbkt[1]{\left[#1\right]}
\newcommand\rbkt[1]{\left(#1\right)}

\newcommand\hinf{\mathcal{H}_{\infty}}
\newcommand\rhinf{\mathcal{RH}_{\infty}}
\newcommand\rtf{\mathcal{R}^{n\times n}}
\newcommand\ccp{\bar{\mathbb{C}}_+}
\newcommand\cop{{\mathbb{C}}_+}
\newcommand\cn{{\mathbb{C}}^n}
\newcommand\cnn{{\mathbb{C}}^{n\times n}}

\newcommand\abs[1]{\left|#1\right|}
\newcommand\rep{{\rm Re}}

\newcommand\imp{{\rm Im}}

\newcommand{\norm}[1]{\left\lVert#1\right\rVert}

\newcommand{\tbt}[4]{\begin{bmatrix}#1&#2\\#3&#4\end{bmatrix}}

\newcommand{\obf}[4]{\begin{bmatrix}#1&#2&#3&#4\end{bmatrix}}

\newcommand{\stbt}[4]{\left[\begin{smallmatrix} #1&#2\\#3&#4\end{smallmatrix}\right]}

\newcommand{\bi}{\begin{itemize}}\newcommand{\ei}{\end{itemize}}
\newcommand{\be}{\begin{equation}}\newcommand{\ee}{\end{equation}}
\newcommand{\bex}{\begin{equation*}}\newcommand{\eex}{\end{equation*}}
\newcommand{\bax}{\begin{align*}}\newcommand{\eax}{\end{align*}}
\newcommand{\bc}{\begin{center}}\newcommand{\ec}{\end{center}}

\newtheorem{theorem}{Theorem}
\newtheorem{proposition}{Proposition}
\newtheorem{definition}{Definition}
\newtheorem{lemma}{Lemma}
\newtheorem{remark}{Remark}
\newtheorem{example}{Example}
\newtheorem{problem}{Problem}
\newtheorem{assumption}{Assumption}
\newtheorem{corollary}{Corollary}
\intervalconfig{soft open fences}

\begin{document}

\title{A Cyclic Small Phase Theorem} 
\author{Chao~Chen, Wei~Chen, Di~Zhao, Jianqi Chen and Li~Qiu,~\IEEEmembership{Fellow,~IEEE} 
\thanks{This work was supported in part by the Research Grants
Council of Hong Kong SAR under the General Research
Fund No.~16201120 and 16203922, by the European Research Council under the Advanced ERC Grant Agreement SpikyControl n.101054323, and by the National Natural Science Foundation of China under the Grants No.~62073003, 72131001, 62103303 and 62088101. (Corresponding author: Li~Qiu)} 
\thanks{Chao~Chen is with the Department of Electrical and Electronic Engineering, University of Manchester, UK. (chao.chen@manchester.ac.uk)}
 \thanks{Wei~Chen is with the Department of Mechanics and Engineering Science, Peking University, Beijing 100871, China. (w.chen@pku.edu.cn)}
 \thanks{Di~Zhao is with the Department of Control Science and Engineering, Tongji University, Shanghai 200092, China. (dzhao925@tongji.edu.cn)}
 \thanks{Jianqi~Chen is with the Center for Advanced Control and Smart Operations, Nanjing University, Suzhou 215163, China. (jqchen@nju.edu.cn)} 
 \thanks{Li~Qiu is with the School of Science and Engineering, The Chinese University of Hong Kong, Shenzhen, China. (qiuli@cuhk.edu.cn)} 
}

\maketitle
\begin{abstract}
This paper introduces a brand-new phase definition called the segmental phase for multi-input multi-output linear time-invariant systems. The underpinning of the definition lies in the matrix segmental phase which, as its name implies, is graphically based on the smallest circular segment covering the matrix normalized numerical range in the unit disk. The matrix segmental phase has the crucial product eigen-phase bound, which makes itself stand out from several existing phase notions in the literature. The proposed bound paves the way for stability analysis of a single-loop cyclic feedback system consisting of multiple subsystems. A cyclic small phase theorem is then established as our main result, which requires the loop system phase to lie between $-\pi$ and $\pi$. The proposed theorem complements a cyclic version of the celebrated small gain theorem. In addition, a generalization of the proposed theorem is made via the use of angular scaling techniques for reducing conservatism. 
\end{abstract}


 \section{Introduction}
 
\IEEEPARstart{T}HE notions of gain and phase act as two supporting and interrelated pillars in classical control theory. The Bode diagram and Nyquist plot as influential and convenient graphical tools are based on gain and phase responses of a single-input single-output (SISO) linear time-invariant (LTI) system. The fruits of the notions are particularly useful in feedback system analysis, including the Nyquist stability criterion, gain/phase margins, and lead/lag compensation techniques.
 
\begin{figure}[htb]
 \centering
 \setlength{\unitlength}{1mm}
 \begin{picture}(100,30)
 \thicklines
 \put(0,20){\vector(1,0){6}} \put(8,20){\circle{4}}
 \put(10,20){\vector(1,0){5}}
 \put(15,15){\framebox(10,10){$P_1$}}
 \put(25,20){\makebox(5,5){$y_1$}}
 \put(30,20){\circle{4}}
 \put(25,20){\vector(1,0){3}}
 \put(30,28){\vector(0,-1){6}}
 \put(30,25){\makebox(5,5){$e_{2}$}}
 \put(32,20){\vector(1,0){5}}
 \put(37,15){\framebox(10,10){$P_2$}}
 \put(47,20){\vector(1,0){3}}
 \put(52,17.5){\makebox(5,5){$\cdots$}}
 \put(59,20){\vector(1,0){3}}
 \put(56,20){\makebox(5,5){$y_{m-1}$}}
 \put(66,20){\vector(1,0){6}}
 \put(64,28){\vector(0,-1){6}}
 \put(64,20){\circle{4}}
 \put(64,25){\makebox(5,5){$e_{m}$}}
 \put(72,15){\framebox(10,10){$P_m$}}
 \put(82,20){\line(1,0){6}} \put(88,20){\line(0,-1){15}}
 \put(88,5){\line(-1,0){80}}
 \put(8,5){\vector(0,1){13}}
 \put(83,20){\makebox(5,5){$y_m$}}
 \put(0,20){\makebox(5,5){$e_1$}}
 \put(32,20){\makebox(5,5){$u_2$}}
 \put(47,20){\makebox(5,5){$y_2$}}
 \put(66,20){\makebox(5,5){$u_{m}$}}
 \put(10,20){\makebox(5,5){$u_1$}} 
 \put(8,10){\makebox(6,10){$-$}}
 \end{picture}
 \vspace{-8.5mm}
 \caption{A single-loop cyclic feedback system consisting of $m$ MIMO LTI subsystems $P_1, P_2, \ldots, P_m\in \rtf$.} \label{fig:feedback}
 \end{figure}

 This paper treats stability analysis of a single-loop cyclic feedback system consisting of $m$ multi-input multi-output (MIMO) LTI subsystems shown in Fig.~\ref{fig:feedback} via a frequency-domain approach. Such a cyclic structure has been widely adopted in the modeling of biochemical and biological systems (see \cite{Tyson:78, Thron:91, Hori:11, Sontag:06, Arcak:06, Scardovi:10} and the references therein). Feedback stability analysis of the structure as a central issue has been well investigated in the literature, with the remarkable secant-gain criteria \cite{Sontag:06, Arcak:06, Scardovi:10, Hamadeh:11, Pates:23, Chaffey:21j}. When the subsystems in Fig.~\ref{fig:feedback} are SISO, the feedback stability can be deduced from the Nyquist criterion which counts the number of encirclements of the critical point ``$-1$'' made by the Nyquist plot of $P(s)\coloneqq P_m(s)P_{m-1}(s)\cdots P_1(s)$. It is often desirable to have this number to be zero, which can be naturally guaranteed from a \textit{small gain} perspective:
 \begin{align*}
 \abs{P(\jw)} = \abs{P_1(\jw)} \abs{P_2(\jw)}\cdots \abs{P_m(\jw)}<1,
 \end{align*}
 or in parallel from a \textit{small phase} perspective:
\bex
 \angle P(\jw) = \angle P_1(\jw) + \angle P_2(\jw) +\cdots + \angle P_m(\jw) \in \interval[open]{-\pi}{\pi}.
\eex
 In either case, the Nyquist plot is strictly contained in a simply connected region, the unit disk or the $\interval[open]{-\pi}{\pi}$-cone, to be away from ``$-1$'', as illustrated by Fig.~\ref{fig:gain_phase}. The two perspectives complement each other and compose a complete story in classical control theory.

 \begin{figure}[htb]
 \vspace{-2.5mm}
 \centering 
\includegraphics[width=3in, trim={0cm 0cm 0cm 0.3cm}, clip]{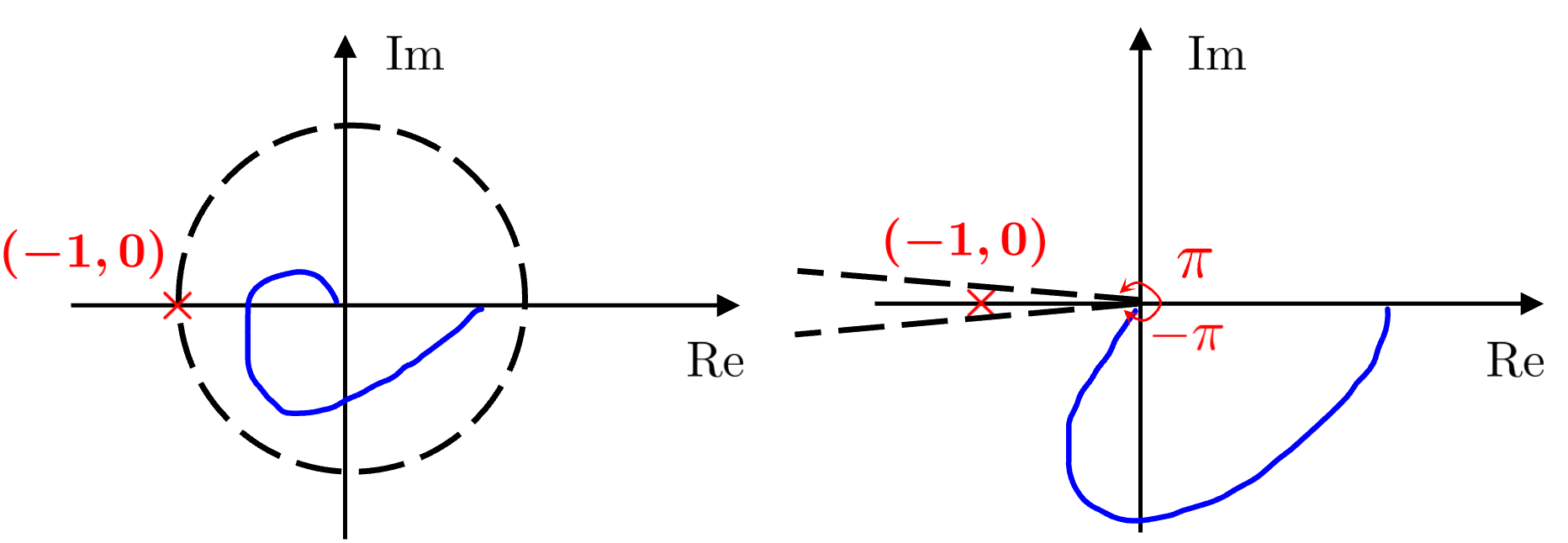} 
\vspace{-1mm}
 \caption{The region (left) from the small gain idea and the region (right) from the small phase idea for Nyquist plots.} \label{fig:gain_phase} 
\end{figure}

When MIMO subsystems are taken into consideration, we likewise aim at establishing the feedback stability from a gain or phase perspective as above based on the generalized Nyquist criterion \cite{Macfarlane:77, Desoer:80, Smith:81} without any encirclement of ``$-1$'' made by the eigenloci of $P(s)$. To this end, one should first know what the notions of gain and phase are for a MIMO system. Over the past half-century, the research on gain-based theories, e.g., the small gain theorem \cite{Zames:66} and $\hinf$ control theory \cite{Zhou:96}, has been flourishing \cite{Jiang:94, Jiang:18,Liu:11}. It is known that the gain of a stable MIMO LTI system represented by its frequency-response matrix is defined by the largest singular value $\overline{\sigma}(\cdot)$ of the matrix. Let $\lambda(\cdot)$ denote one of the eigenvalues of a matrix and $\lambda(P(\jw))$ is referred to as an eigenlocus of $P(s)$. The gain definition gives birth to the monumental \emph{small gain} condition for a cyclic loop:
\bex
\abs{\lambda(P(\jw))} \leq \overline{\sigma}(P_1(\jw)) \overline{\sigma}(P_2(\jw)) \cdots \overline{\sigma}(P_m(\jw)) <1.
\eex
 The condition above as an extension of the SISO gain case is also termed the \textit{product eigen-gain bound} for matrices. 

In comparison to the prosperous gain-based theory, the phase counterpart is however very much under-developed. For a long time there has even been a distinct lack of consensus on this counterpart among researchers. What is the phase of a matrix? What is the phase of a MIMO system? What is the formulation of a MIMO small phase theorem? These are the fundamental questions constantly asked by prominent control researchers in the 1980-1990s -- various phasic notions were developed for different purposes, such as the ``principal phase'' \cite{Postlethwaite:81}, gain/phase integral relation \cite{Anderson:88, Freudenberg:88, Chen:98}, phase uncertainty \cite{Owens:84, Tits:99}, and phase margin \cite{bar-on:90}. In addition, there were three phase-related qualitative notions: the positive realness \cite{Anderson:73}, negative imaginariness \cite{Lanzon:08, Petersen:10, Lanzon:22, Zhao:22_NI} and relaxation-type dynamics \cite{Willems:72_2, Pates:19, Chaffey:23C}. Very recently, these questions were also addressed with the thriving fruits upon the notion ``sectorial phase'' -- on matrices \cite{Wang:20, Zhao_matrix:22}, MIMO LTI systems and networks \cite{Chen:21,Zhao:22, Liang:24, ChenJ:23}, and nonlinear systems \cite{Chen:20j}. In these references, the ``sectorial phase'' has shown its advantages in studying the feedback loop of two subsystems. 
However, none of the existing phase notions above, to the best of our knowledge, are suitable for stability analysis of cyclic feedback systems with at least three components, i.e., $m\geq 3$. The crux of the issue lies in a lack of an appropriate matrix phase definition $\Psi(\cdot)$ having the following vital \textit{product eigen-phase bound} for $m\geq 3$ matrices:
\begin{align*}
 \angle{\lambda(P(\jw))} \in \Psi(P_1(\jw))+ \cdots + \Psi(P_m(\jw)) \subset \interval[open]{-\pi}{\pi}.
\end{align*}
 Such a bound recovers the SISO phase case. More importantly, it acts as a significant enabling instrument in the formulation of a small phase theorem for a cyclic loop, analogously to the role of the product eigen-gain bound, as mentioned earlier, played in the small gain theorem.

Let us go deeper into two representative references \cite{Postlethwaite:81, Chen:21} in respect of the product eigen-phase bound. The pioneering work \cite{Postlethwaite:81} proposes the ``principal phase'' of a matrix based on its polar decomposition \cite{Horn:59}. The product eigen-phase bound therein \cite[Th.~2]{Postlethwaite:81} only holds for one invertible matrix ($m=1$) provided that the spread of ``principal phase'' is less than $\pi$. This leads to an extra condition in the small ``principal phase'' result \cite[(b) in Th.~4]{Postlethwaite:81}. The recent paper \cite{Chen:21} proposes the ``sectorial phase'', aka canonical angle \cite{Furtado:01}, of a sectorial matrix based on the numerical range and sectorial decomposition. The product eigen-phase bound therein \cite[Lem.~2.4]{Chen:21} holds for two sectorial matrices, thereby leading to the successful small ``sectorial phase'' theorem \cite[Th.~4.1]{Chen:21} involving two sectorial subsystems. The theorem however cannot be extended to a cyclic loop with more than two components due to the intrinsic limitation of ``sectorial phase''. Besides, studying only the sectorial-type of matrices/systems may be considered as another obvious limitation. 

Motivated in part by the importance of phase in classical control theory and by the ``principal phase'' and ``sectorial phase'' above, we expect an alternative notion of phase for matrices and MIMO systems having the product eigen-phase bound for any $m$ components, and then exploit this notion in stability analysis of cyclic feedback systems.

In this paper, we first propose a brand-new phase notion called the \textit{segmental phase} for matrices and MIMO LTI systems. The proposed matrix segmental phase is graphically based on the normalized numerical range, a simply connected region contained in the closed unit disk, and is defined through the smallest circular segment of the disk covering the region. In particular, the matrix segmental phase has the crucial {product eigen-phase bound} as expected. We then formulate a \textit{small phase theorem} for stability analysis of a cyclic loop involving semi-stable MIMO subsystems in Fig.~\ref{fig:feedback}, which requires the loop system phase to lie inside the $\interval[open]{-\pi}{\pi}$-cone, i.e., 
$\Psi(P_1(\jw))+ \cdots + \Psi(P_m(\jw)) \subset \interval[open]{-\pi}{\pi}$.
The proposed theorem generalizes the classical SISO phase case and serves as a counterpart of the MIMO small gain theorem. A mixed small gain/phase theorem is further established for practical use via frequency-wise gain/phase conditions. Finally, an angular scaling technique is proposed for reducing conservatism when the small phase theorem is applied to a cyclic loop. The technique is particularly targeted at the scenario where the loop consists of known subsystems and phase-bounded uncertain subsystems simultaneously. 

This paper has substantial contributions beyond the authors' conference paper \cite{Chen:22IFAC} whose focal point lies on stable MIMO systems without zeros on the imaginary axis. The scope of the current paper covers more general subjects: The small phase theorem is now applicable to possibly semi-stable systems. The mixed small gain/phase theorem stated in a frequency-wise manner is new. Additionally, the proposed angular scaling technique for reducing conservatism of the main result is new. 

The rest of this paper is structured as follows. In Section~\ref{sec:problem}, we formulate the main problem: find a phasic stability condition of cyclic feedback systems. In Section~\ref{sec:matrix}, we define the segmental phase of a matrix based on the normalized numerical range. A comparison is made between the segmental phase with existing phase definitions. With the established mathematical underpinning, the segmental phase of a MIMO system is developed in Section~\ref{sec:systems}. Section~\ref{sec:small_phase} is dedicated to small phase theorems for stability analysis of cyclic feedback systems. For reducing conservatism the theorem, in Section~\ref{sec:generalized_small_phase} an angular scaling technique for dealing with cyclic loops is proposed. Section~\ref{sec:conclusion} concludes this paper. 

 \textit{Notation}: 
 The notation used in this paper is standard.
 For intervals $\interval{a}{b}$ and $\interval{c}{d}$ with $a\leq b$ and $c\leq d$, define the Minkowski addition and subtraction as $\interval{a}{b}+\interval{c}{d}=\interval{a+c}{b+d}$ and $\interval{a}{b}-\interval{c}{d}=\interval{a-d}{b-c}$ by convention, respectively. The interval $[a, a]$ as a singleton is shortened to $a$. The argument of an extended complex number $z\in \bar{\mathbb{C}}\coloneqq\mathbb{C}\cup \bbkt{\infty}$ is denoted by $\angle z$, and $z$ has no argument $\angle z= \emptyset$ if $z=0$ or $z=\infty$. The conventions that $\emptyset + \emptyset = \emptyset$, $\inf \emptyset = \infty$ and $\sup \emptyset = -\infty$ are adopted. Let $\lambda_i(A)\in \mathbb{C}$ denote the $i$-th eigenvalue of a matrix $A\in \cnn$, where $i=1,2, \ldots, n$. Denote $\cop$ and $\ccp$ as the open and closed complex right half-planes, respectively. For a set $\mathcal{S}$ and $c\in \mathbb{C}$, let $c\mathcal{S} \coloneqq \{ c z\mid z\in \mathcal{S} \}$. Denote by $\mathcal{R}^{n\times n}$ the set of $n\times n$ real-rational proper transfer function matrices. Let $\rhinf^{n\times n}$ denote the subset of $\mathcal{R}^{n\times n}$ consisting of transfer function matrices with no poles in $\ccp$. A system $P\in \rtf$ is called semi-stable if it has no poles in $\cop$; it is called stable if $P\in \rhinf^{n\times n}$. 

\section{Problem Formulation}\label{sec:problem}

Consider a cyclic feedback system shown in Fig.~\ref{fig:feedback}, where $P_k\in \rtf$ are MIMO systems, $e_k$ are external signals, and $u_k$ and $y_k$ are internal signals for $k=1, 2, \ldots, m$. Denote $e\coloneqq \obf{e_1^\top}{e_2^\top}{\cdots}{e_m^\top}^\top$ and $u\coloneqq \obf{u_1^\top}{u_2^\top}{\cdots}{u_m^\top}^\top$. Our major interest is the stability of the cyclic feedback system defined as follows: 
\begin{definition}\label{def:internal_stability}
 A cyclic feedback system in Fig.~\ref{fig:feedback} is said to be \emph{stable} if the transfer matrix (i.e., the mapping $e\mapsto u$)
 \begin{align}\label{eq:io_map_stability}
 \left[\begin{matrix} I& 0 & \cdots & 0& P_m \\
 -P_1 & I & 0 & \cdots & 0 \\
 \vdots & \ddots & \ddots & \ddots & \vdots \\
 0 & \cdots & -P_{m-2} & I & 0 \\
 0 & \cdots & 0 & -P_{m-1} & I \\
 \end{matrix}\right]^{-1}
 \end{align}
 belongs to $\rhinf^{mn\times mn}$. 
\end{definition}

Denote by $p_k$ the number of unstable poles of $P_k(s)$ for $k=1, 2, \ldots, m$. A cyclic feedback system is said to have \emph{no unstable pole-zero cancellation} if the number of unstable poles of $P_{m}(s)P_{m-1}(s)\cdots P_1(s)$ is equal to $\sum_{k=1}^m p_k$. Throughout, all cyclic feedback systems are reasonably \emph{assumed to be free of unstable pole-zero cancellations}. Under such an assumption, Definition~\ref{def:internal_stability} has an equivalent and simple characterization in terms of $\rbkt{I+P_mP_{m-1}\cdots P_1}^{-1}$ elaborated as follows:

\begin{lemma}\label{lem:feedback_no_cancellation}
 A cyclic feedback system in Fig.~\ref{fig:feedback} is stable if and only if it has no unstable pole-zero cancellation and \be \label{eq:stability}
 \rbkt{I+P_m P_{m-1}\cdots P_1}^{-1} \in \rhinf^{n\times n}.
 \ee
 \end{lemma}
 \begin{proof}
 We can follow similar arguments in the proof of \cite[Th.~5.7]{Zhou:96} and thus only one core step needs to be pointed out.
 Under the pole-zero condition, following some tedious calculations, one can check that the state matrix $\bar{A}$ of the minimal realization of \eqref{eq:stability} is equal to the state matrix $\tilde{A}$ of the minimal realization of in \eqref{eq:io_map_stability}, i.e., $ \bar{A}=\tilde{A}$. 
 \end{proof}

One may attempt to analyze stability of a cyclic loop by leveraging the generalized Nyquist criterion \cite{Desoer:80} directly for $P_m P_{m-1}\cdots P_1$. In many applications, some of subsystems may not be precisely known and are oftentimes described by appropriate \emph{uncertain} sets, while the remaining ones are assumed to be {known}. Requiring eigenlocus information of uncertain systems is {unrealistic}. Moreover, eigenloci do not give reliable information: they are neither good robust stability indicators nor good robust performance indicators \cite[Sec.~2.1]{Green:12}. Robust control theory is targeted at such a scenario, where the monumental small gain theorem \cite{Zames:66} has been one of the most important tools. A prerequisite for connecting the gain-based analysis to a feedback loop is to characterize uncertain systems by appropriate gain-bounded sets.
 
Denote by $\mathcal{B}_{\delta}^\nn$ the following set of gain-bounded systems: 
\bex
\mathcal{B}_\delta^\nn \coloneqq \bbkt{P\in \rhinf^{n\times n}\mid \overline{\sigma}(P(\jw))\leq \delta(\omega), \omega \in \interval{0}{\infty}},
\eex 
where $\delta\colon \interval{0}{\infty} \to \interval[open right]{0}{\infty}$ represents a frequency-wise finite gain bound. Let $\mathcal{K}\subset \bbkt{1, 2, \ldots, m} $ denote the index set of uncertain systems in a cyclic loop. The complement of $\mathcal{K}$ is denoted by $\mathcal{K}^\prime \coloneqq \bbkt{1, 2, \ldots, m} \setminus \mathcal{K}$. A direct application of the small gain theorem \cite[Th.~9.1]{Zhou:96} yields the following result: 
\begin{lemma}\label{lem:small_gain_cyc}
 Let $P_1, P_2, \ldots, P_m\in \rhinf^{n\times n}$ and assume that $P_k(s)\in \mathcal{B}_{\delta_k}^\nn$ for $k\in \mathcal{K}$, where $\delta_k\colon \interval{0}{\infty} \to \interval[open right]{0}{\infty}$. The cyclic feedback system is stable if for all $\omega\in \interval{0}{\infty}$,
\be\label{eq:small_gain_cyc}
\prod_{k\in \mathcal{K}} \delta_k(\omega)\prod_{k\in \mathcal{K}^\prime}\overline{\sigma}(P_k(\jw)) <1.
 \ee
\end{lemma}

 The gain-based condition \eqref{eq:small_gain_cyc} is composed of two parts. For those uncertain systems, their available information of gain bounds $\delta_k(\cdot)$ can only be adopted. For the remaining ones, their exact gains $\overline{\sigma}(\cdot)$ can be computed, and hence it becomes possible to reduce the conservatism of \eqref{eq:small_gain_cyc} via widely-adopted gain-scaling and loop-shaping techniques (see \cite[Ch.~11]{Zhou:96}).

We yearn for a parallel to the small gain theorem (Lemma~\ref{lem:small_gain_cyc}) from a phasic perspective. To this end, we target the following main problem:
\begin{problem}[Feedback stability]\label{problem}
Consider $P_1, P_2, \ldots, P_m$ $\in \rtf$ and assume that $P_k$ belongs to a certain set of ``phase-bounded'' uncertain systems for $k\in \mathcal{K}$. Find a phasic condition for stability of the cyclic feedback system.
\end{problem}

For SISO systems, the notion of phase is standard and thereby Problem~\ref{problem} can be completely solved using the Nyquist stability criterion as stated in the introduction. For MIMO systems, Problem~\ref{problem} is however quite \emph{nontrivial}. To resolve this, our first need is a suitable MIMO phase definition which has been recognized as a significant issue among several generations of control researchers \cite{Postlethwaite:81, Owens:84,Freudenberg:88, Anderson:88, Chen:98, Tits:99,Chen:21}. The bottleneck behind the issue toward Problem~\ref{problem} lies in the {lack of a suitable matrix phase definition possessing the critical product eigen-phase bound}, as pointed out in the introduction. This drives us to start with a simpler problem regarding to complex matrices.
\begin{problem}[Matrix invertibility]\label{problem_matrix}
 Find a phasic condition for invertibility of the matrix $I+A_mA_{m-1}\cdots A_1$, where $A_1, A_2, \ldots, A_m \in \mathbb{C}^{n\times n}$ with $A_k$ belongs to a certain set of ``phase-bounded'' uncertain matrices for $k\in \mathcal{K}$.
\end{problem}

 Solving Problem~\ref{problem_matrix} serves as an intermediate but crucial step toward Problem~\ref{problem} for frequency-domain analysis of MIMO systems in Section~\ref{sec:small_phase}. Our solution will be based on the development of a new phase definition for matrices and systems. To highlight the definition itself, for the time being let us assume that in Sections~\ref{sec:matrix}-\ref{sec:small_phase} all the matrices and systems under consideration are known, i.e., $\mathcal{K}=\emptyset$. We thus postpone the issue of uncertain components until Section~\ref{sec:generalized_small_phase} in which Problems~\ref{problem} and \ref{problem_matrix} will be completely resolved.

\section{The Segmental Phase of a Matrix}\label{sec:matrix} 
\subsection{Definition of the Segmental Phase}
In this subsection, we establish a new matrix phase definition called the segmental phase, utilizing the normalized numerical phase of the matrix as studied in \cite{Auzinger:03, Lins:18}. The proposed definition yields the product eigen-phase bound, based on which a matrix small phase theorem is then obtained. 
 
For a matrix $A\in \cnn$, the \textit{normalized numerical range} $\mathcal{N}(A)$ is defined by
\be\label{eq:matrix_NNR}
\mathcal{N}(A)\coloneqq \bbkt{ \frac{x^* A x}{\abs{x}\abs{Ax}} \in \mathbb{C}~\bigg|~ 0\neq x\in \cn, Ax\neq 0},
\ee
 where $\abs{x}\coloneqq \sqrt{x^* x}$ and $x^*$ denote the Euclidean norm and the conjugate transpose of $x$, respectively. The normalized numerical range $\mathcal{N}(A)$ as a subset of the complex plane is simply connected \cite[Prop.~2.1]{Lins:18}. By the Cauchy-Schwarz inequalities, $\mathcal{N}(A)$ is contained in the closed unit disk $\bar{\mathbb{D}}$ which manifests that gain information of $A$ has been normalized:
 \begin{align}\label{eq: rotation_NNR}
 \mathcal{N}(c A) = e^{j\angle c}\mathcal{N}(A)
 \end{align}
 for any nonzero $c\in \mathbb{C}$. Additionally, $\mathcal{N}(A)$ intersects with the unit circle only at the nonzero normalized eigenvalues of $A$, i.e., at $\frac{\lambda_i(A)}{\abs{\lambda_i(A)}}$. A typical normalized numerical range is shown as the grey area in Fig.~\ref{fig:seg_phase}. Another quite obvious property of $\mathcal{N}(A)$ is that it is invariant to unitary similarity transformation, i.e., $\mathcal{N} (U^*A U)=\mathcal{N}(A)$ for every unitary matrix $U\in \cnn$.

 \begin{example}\label{example1}
The normalized numerical ranges of the following classes of matrices can be obtained easily in Fig.~\ref{fig:matrix_exmp}.
 \begin{enumerate}
 \renewcommand{\theenumi}{\textup{(\roman{enumi})}}\renewcommand{\labelenumi}{\theenumi}
 \item \label{item:exmp1}
 If $A=ke^{ j \alpha}I$, a scalar matrix, where $k>0$ and $\alpha \in \interval[open right]{-\pi}{\pi}$, then $\mathcal{N}(A)$ is a singleton at $e^{j \alpha}$.
 \item \label{item:exmp2}
 If $A$ is a positive definite matrix, then $\mathcal{N}(A)$ is a line segment connecting $\frac{2\sqrt{\kappa(A)}}{\kappa (A)+1}$ to $1$, where $\kappa(A)\coloneqq \frac{\overline{\lambda}(A)}{\underline{\lambda}(A)}$ denotes the condition number of $A$. 
 \item \label{item:exmp3} If $A$ is a unitary matrix, then $\mathcal{N}(A)$ is a polygon with the eigenvalues as vertices.
 \item \label{item:exmp3}
 If $A$ is a nonzero nilpotent Jordan matrix, then $\mathcal{N}(A)$ is the whole open unit disk.
 \end{enumerate}
 \end{example}

To dig up phase information of $A$ contained in $\mathcal{N}(A)$, we tailor \textit{a smallest circular segment}, i.e., one with a \emph{shortest arc edge}, of the unit disk to cover $\mathcal{N}(A)$. As shown in Fig.~\ref{fig:seg_phase}, the grey region $\mathcal{N}(A)$ is covered by the red-bordered segment. 
\begin{definition}\label{def: segphase_matrix}
The \textit{segmental phase} of $A\in \cnn$ is defined to be multi-interval-valued, with each~{interval} specified by the arc edge of a smallest segment covering $\mathcal{N}(A)$ in Fig.~\ref{fig:seg_phase}:
\be\label{eq:seg_phase_def}
 \Psi(A)\coloneqq \interval{\underline{\psi}(A)}{\overline{\psi}(A)},
\ee
 where $\underline{\psi}(A)$ (resp. $\overline{\psi}(A)$) is given by the {argument} of the lower (resp. upper) endpoint of the arc edge. The \textit{phase center} of $A$ is multi-valued given by
 \be\label{eq:seg_phase_center_def}
 \gamma^\star(A)\coloneqq \textstyle\frac{1}{2}\rbkt{\underline{\psi}(A)+\overline{\psi}(A)},
 \ee
 with each element ${\gamma}^\star \in \gamma^\star(A)$ satisfying that $\gamma^\star \in \interval[open right]{-\pi}{\pi}$.
\end{definition}

It holds that $\mathcal{N}(A)=\emptyset$ if and only if $A=0$. Hence a zero matrix has no phase $\Psi(A)=\emptyset$. For notational brevity, throughout this paper we adopt $\Psi(A)$ in \eqref{eq:seg_phase_def} for representing $\{ \interval{\underline{\psi}_r(A)}{\overline{\psi}_r(A)} \mid r\in \mathbb{Z}_+\}$ and $\gamma^\star(A)$ in \eqref{eq:seg_phase_center_def} for $\{\gamma^\star_r(A) \mid r\in \mathbb{Z}_+\}$, even when $r=1$ only. A \emph{phase-interval selection $\Psi\in \Psi(A)$} will always be clarified for non-singleton $\Psi(A)$ once necessary. Note that each phase interval ${\Psi}\in \Psi(A)$ has the same length bounded by $2\pi$. 

\begin{figure}[htb]
 \centering
 \vspace{-2mm}
 \includegraphics[width=2.5in, trim={1.8cm 1.7cm 1.5cm 1.2cm}, clip]{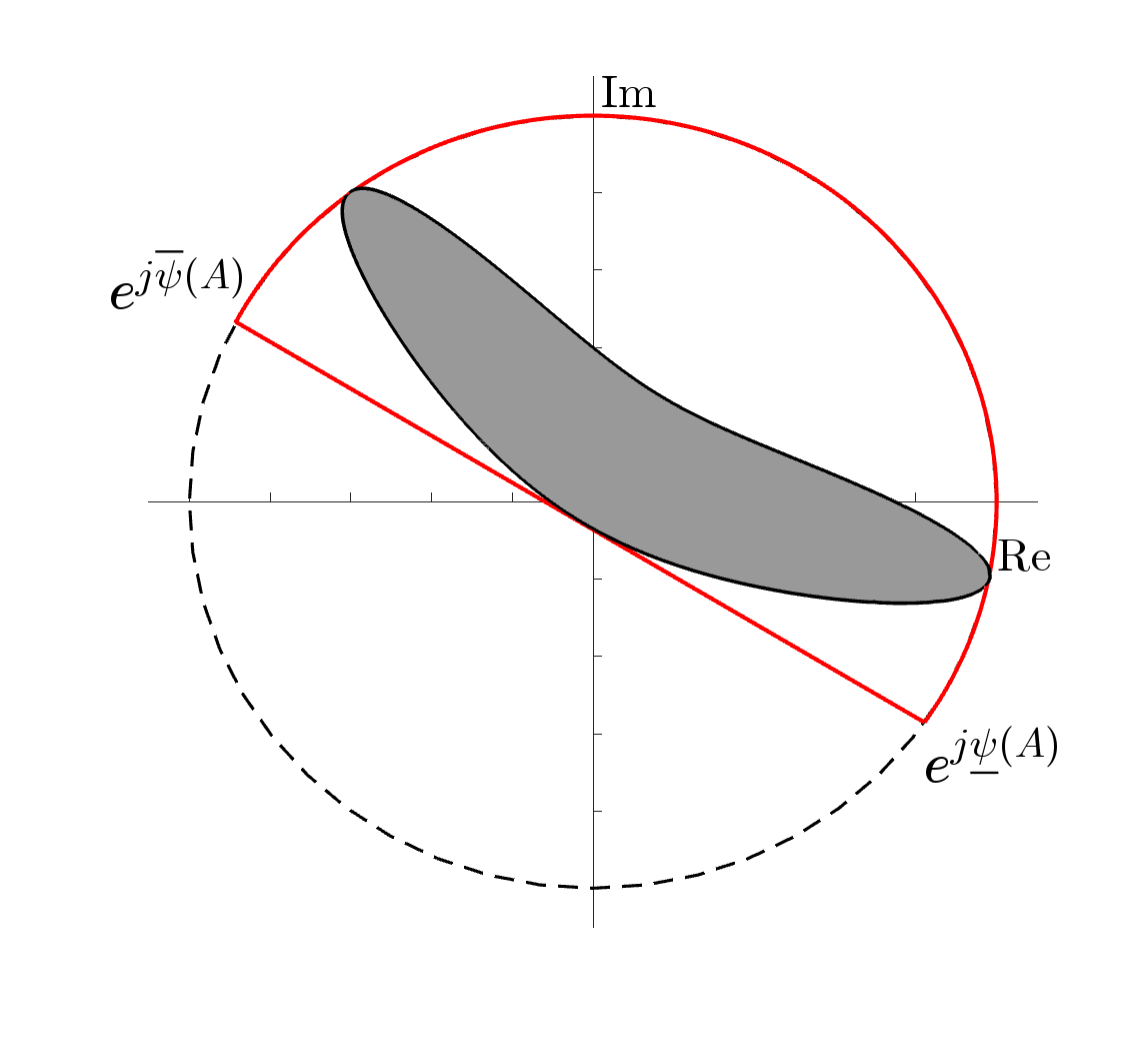}
 \vspace{-3mm}
 \caption{A graphical illustration of the segmental phase of $A=\stbt{-3+j4}{1}{0}{{1}/(5+j)}$. The grey area is the normalized numerical range $\mathcal{N}(A)$ contained in the unit disk. The red-bordered smallest circular segment of the disk is adopted to cover $\mathcal{N}(A)$. The lower and upper endpoints of the arc of the segment are respectively given by $e^{j\underline{\psi}(A)}$ and $e^{j\overline{\psi}(A)}$, which defines the segmental phase interval $\Psi(A)=\interval{\underline{\psi}(A)}{\overline{\psi}(A)}=\interval[scaled]{\frac{-34.2\pi}{180}}{\frac{152.2\pi}{180}}$.} \label{fig:seg_phase} 
\end{figure} 

Let us expand the multi-interval-valued issue of $\Psi(A)$. A segment covering $\mathcal{N}(A)$ with the shortest arc edge is not {necessarily unique}, which makes the phase interval $\Psi(A)$ non-unique. One can easily construct a unitary matrix $A$ with eigenvalues evenly distributed on the unit circle and hence $\mathcal{N}(A)$ being a regular polygon. In this case, the number of such smallest segments covering $\mathcal{N}(A)$ is exactly equal to $n$. Another extreme situation is when $A$ is a nonzero nilpotent Jordan matrix as in Example~\ref{example1}\ref{item:exmp3}. In this case, the smallest segment degenerates to the unit disk. Hence $\Psi(A)$ can be an arbitrary $2\pi$-interval. The non-uniqueness is sometimes a cause of trouble. The first good news is given by the following proposition, in which a minor (major) segment is one strictly smaller (larger) than a semi-disk.
 
\begin{proposition}\label{prop: sectorial}
For invertible $A\in \cnn$, if $\mathcal{N}(A)$ 
is contained in a minor segment, then $\Psi(A)$ is a singleton.
\end{proposition}
\begin{proof}
 By \cite[Prop.~2.1]{Lins:18}, $\mathcal{N}(A)$ is a closed set when $A$ is invertible. By hypothesis, any minor segment $\mathcal{S}$ covering $\mathcal{N}(A)$ excludes the origin. Assume that there are two different smallest minor segments $\mathcal{S}_1$ and $\mathcal{S}_2$ covering $\mathcal{N}(A)$, and the arc interval of $\mathcal{S}_1$ and that of $\mathcal{S}_2$ have the same length being less than $\pi$. The chord of $\mathcal{S}_1$ and that of $\mathcal{S}_2$ must intersect at some unique nonzero $z\in \bar{\mathbb{D}}$. Then, one can always obtain a new minor segment $\mathcal{S}_3$ by constructing the unique shortest chord passing thorough $z$, that is, the one perpendicular to the line segment connecting $0$ and $z$. The new chord lies between the chord of $\mathcal{S}_1$ and that of $\mathcal{S}_2$, and thus $\mathcal{S}_3$ also covers $\mathcal{N}(A)$. The arc interval of $\mathcal{S}_3$ will be smaller than that of $\mathcal{S}_1$ and that of $\mathcal{S}_2$, which implies that $\mathcal{S}_1$ and $\mathcal{S}_2$ cannot be the smallest segments covering $\mathcal{N}(A)$. This shows a contradiction. Thus the smallest segment is unique and $\Psi(A)$ is a singleton. 
\end{proof}
 
Note that $\angle\lambda_i(A)$ is defined modulo $2\pi$. For any interval $\Psi\in \Psi(A)$, one can always find correspondingly a $2\pi$-interval for $\angle\lambda_i(A)$ such that by \eqref{eq:seg_phase_def}, $\Psi$ provides a bound of $\angle \lambda_i(A)$ from above and below, i.e., $\angle \lambda_i(A) \in \Psi$. For some classes of matrices, the segmental phases have simple expressions in terms of eigenvalues: 
 \begin{example}\label{example2}
 We revisit the case studies in Example~\ref{example1}. 
 \begin{enumerate}
 \renewcommand{\theenumi}{\textup{(\roman{enumi})}}\renewcommand{\labelenumi}{\theenumi}
 \item \label{item:prop_1}
 If $A=ke^{ j \alpha}I$, where $k>0$ and $\alpha \in \interval[open right]{-\pi}{\pi}$, then $\gamma^\star(A)=\alpha$ and $\Psi(A)=\alpha$.
 \item \label{item:prop_3}
 If $A$ is a positive definite matrix, then $\gamma^{\star}(A)=0$ and
 \bex 
 \Psi(A)= \textstyle\interval[scaled]{\displaystyle -\arccos\frac{2\sqrt{\kappa(A)}}{\kappa(A)+1}}{\arccos\frac{2\sqrt{\kappa(A)}}{\kappa(A)+1}}.
 \eex
 \item \label{item:prop_4} If $A$ is a unitary matrix, the longest side of $\mathcal{N}(A)$ divides the unit disk into two segments, one covering $\mathcal{N}(A)$ and one does not. The former is a smallest segment covering $\mathcal{N}(A)$ and its corners are two neighboring eigenvalues $e^{j\underline{\psi}}$ and $e^{j\overline{\psi}}$ of $A$ with a largest circular gap. The segmental phase of $A$ is given by $\Psi(A) = [\underline{\psi}, \overline{\psi}]$.
 
 \item \label{item:nilpotent} If $A$ is a nilpotent Jordan matrix, then $\Psi(A)=[\gamma^\star(A)-\pi, \gamma^\star(A)+\pi]$, where $\gamma^\star(A) \in \interval[open right]{-\pi}{\pi}$ is arbitrary.
 \end{enumerate}
 \end{example}

\begin{figure}[htb]
 \vspace{-2.5mm}
 \centering
 \begin{subfigure}[b]{0.24\textwidth}
 \centering
 \includegraphics[width=1.72in, trim={2.4cm 2cm 1.6cm 1.2cm}, clip]{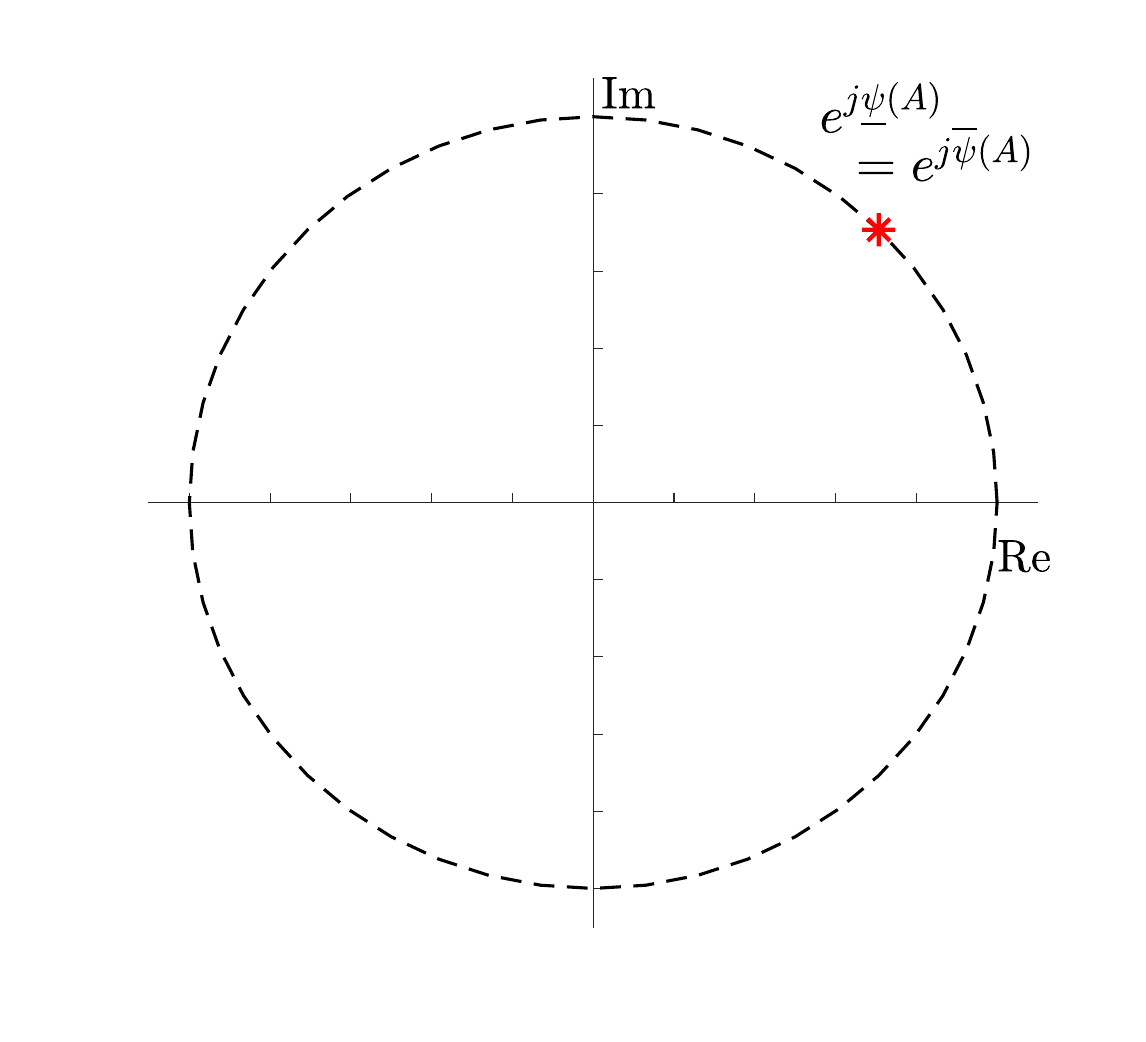} 
 \vspace{-6mm}
 \caption{} 
 \end{subfigure}\hfill
 \begin{subfigure}[b]{0.24\textwidth}
 \centering
 \includegraphics[width=1.72in,trim={2.4cm 2cm 1.6cm 1.2cm}, clip]{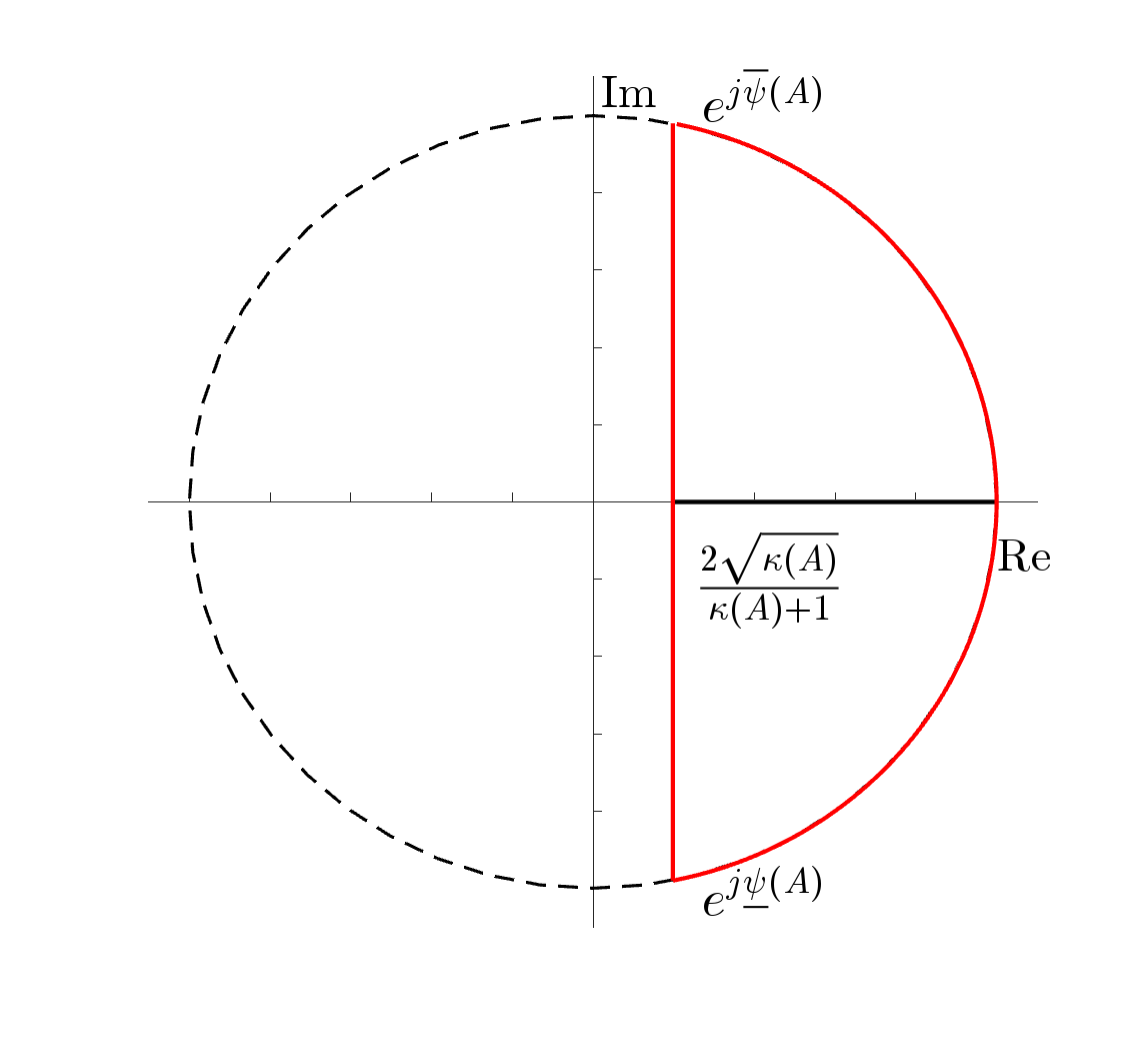}
 \vspace{-6mm}
 \caption{} 
 \end{subfigure}
 \begin{subfigure}[b]{0.24\textwidth}
 \centering
 \includegraphics[width=1.72in, trim={2.4cm 1.6cm 1.6cm 1.2cm}, clip]{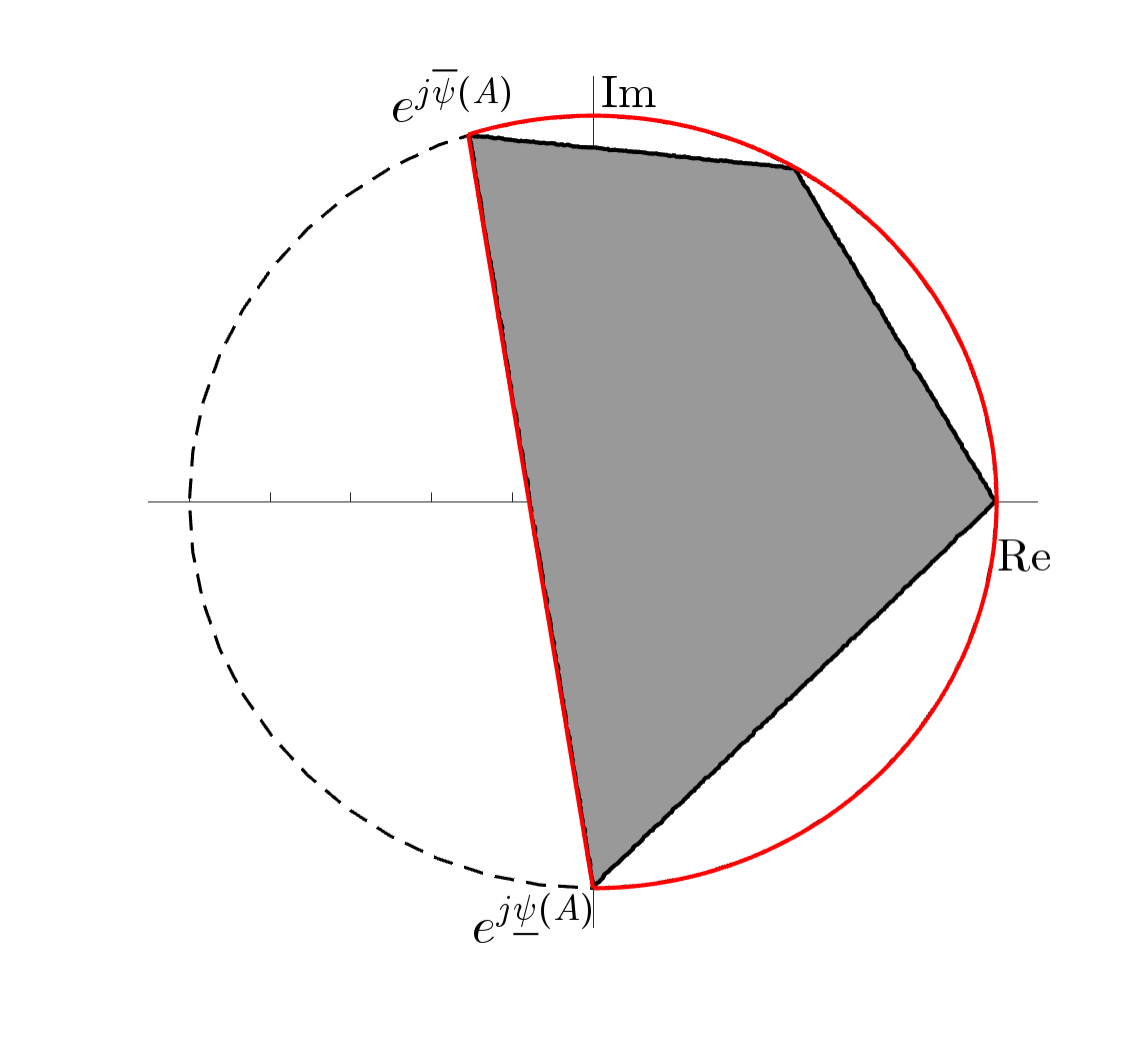}
 \vspace{-6.5mm}
 \caption{} 
 \end{subfigure}\hfill
 \begin{subfigure}[b]{0.24\textwidth}
 \centering
 \includegraphics[width=1.72in, trim={2.4cm 1.6cm 1.6cm 1.2cm}, clip]{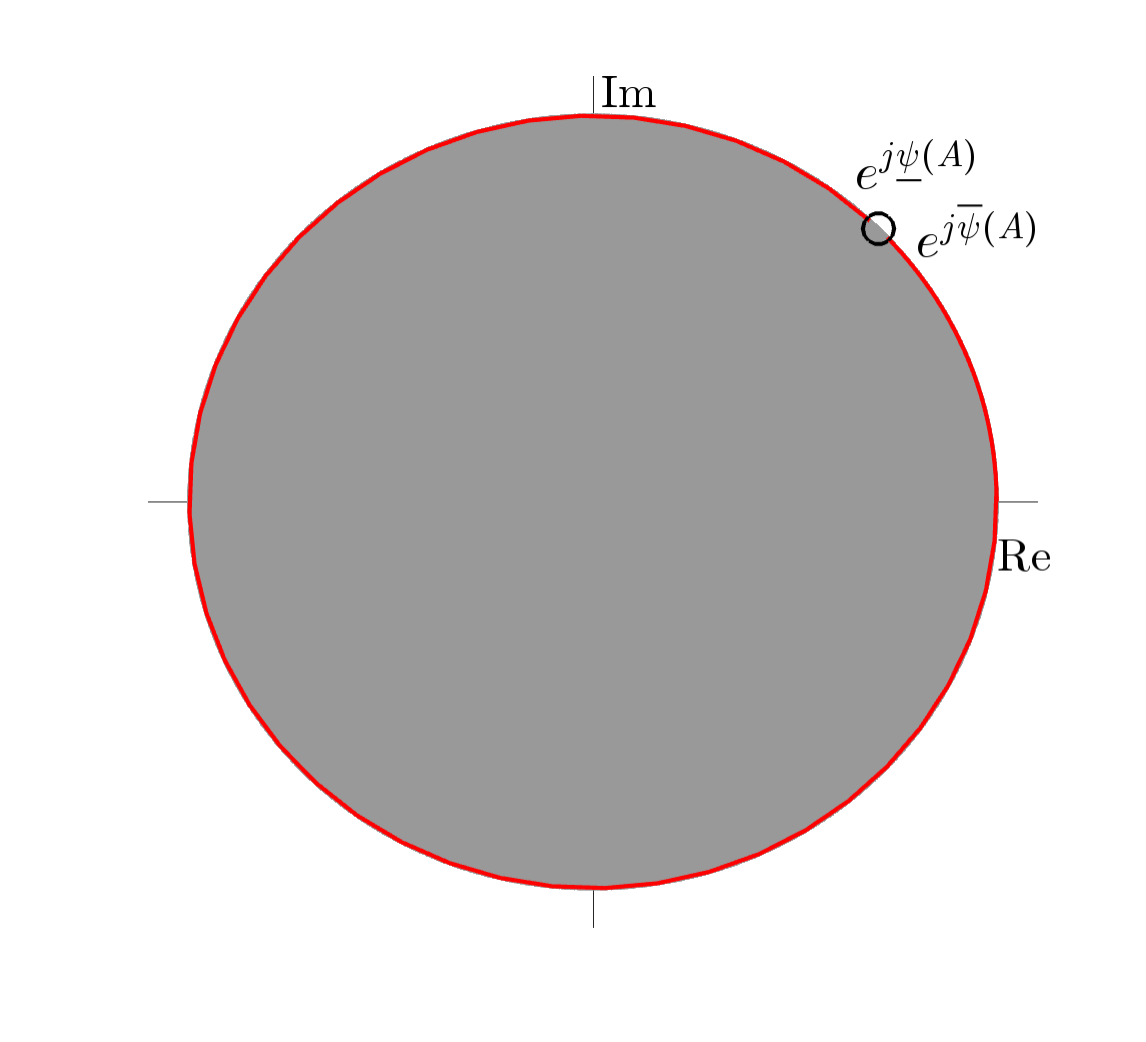} 
 \vspace{-6.5mm}
 \caption{} 
 \end{subfigure}
 \vspace{-1mm}
 \caption{$\mathcal{N}(A)$ and $\Psi(A)$ for (a) scalar identity $A=100e^{j\frac{\pi}{4}}I$, (b) positive definite $A=\mathrm{diag}(1, 20, 100)$, (c) unitary $A=\mathrm{diag}(-j, 1, e^{j\frac{\pi}{3}}, e^{j\frac{3\pi}{5}})$ and (d)~nilpotent Jordan $A=\stbt{0}{1}{0}{0}$.}\label{fig:matrix_exmp} 
\end{figure}

We proceed to showing that the segmental phase is particularly useful in studying product of multiple matrices. The non-uniqueness provides \emph{additional flexibility} in our analysis. The crucial \textit{product eigen-phase bound} is elaborated as follows:

\begin{theorem}\label{lem:eigen_phase}
Let $A_1, A_2, \ldots, A_m\in \cnn$. For each interval $\Psi_k \in \Psi(A_k)$ with $k=1, 2, \ldots, m$, the arguments of the eigenvalues of $A_mA_{m-1}\cdots A_1$ can be chosen accordingly such that
 \bex
 \angle \lambda_i(A_mA_{m-1}\cdots A_1) \in \sum_{k=1}^{m} \Psi_k 
 \eex
 for all $i=1, 2, \ldots, n$, where $\lambda_i(A_mA_{m-1}\cdots A_1) \neq 0$.
 \end{theorem}
 
The full proof of Theorem~\ref{lem:eigen_phase} is provided in Appendix~\ref{appendix:phase}, after we establish the connection between the segmental phase and the singular angle \cite[Sec.~23.5]{Wielandt:67}, an old but less known concept. The essential property showcased in Theorem~\ref{lem:eigen_phase} holds regardless of any selection $\Psi_k$, which enables the use of the segmental phase in stability analysis of cyclic feedback systems and forms the nucleus of the phase study. Following immediately from Theorem~\ref{lem:eigen_phase}, we can establish a \emph{cyclic small phase theorem for matrices} as our first answer to Problem~\ref{problem_matrix}.

\begin{theorem}\label{thm:matrix_segmental_phase}
 Let $A_1, A_2, \ldots, A_m \in\cnn$. The matrix $I+A_mA_{m-1}\cdots A_1$ is invertible if there exist an integer $l\in \mathbb{Z}$ and $\Psi_k \in \Psi(A_k)$ with $k=1, 2, \ldots, m$ such that
 \be\label{eq:matrix_segmental_phase}
 \sum_{k=1}^{m} \Psi_k \subset \interval[open]{-\pi}{\pi} + 2\pi l.
 \ee
\end{theorem}
\begin{proof}
By hypothesis, according to \eqref{eq:matrix_segmental_phase} and Theorem~\ref{lem:eigen_phase}, for specific $l$ and $\Psi_k$ with $k=1, 2, \ldots, m$, it holds that \bex {\angle \lambda_i(A_mA_{m-1}\cdots A_1)}\in \textstyle \sum_{k=1}^{m} \Psi_k \subset \interval[open]{-\pi}{\pi } + 2\pi l \eex for all $i=1, 2, \ldots, n$, where $\lambda_i(A_mA_{m-1}\cdots A_1)\neq 0$.
 This implies that $\det\rbkt{I+A_mA_{m-1}\cdots A_1}\neq 0$ and $I+A_mA_{m-1}\cdots A_1$ is invertible.
\end{proof}
 
 By \eqref{eq:matrix_segmental_phase} and Proposition~\ref{prop: sectorial}, one may observe that in Theorem~\ref{thm:matrix_segmental_phase}, \emph{at most one matrix} is permitted to have multi-valued phase intervals. Otherwise, the length of the interval $\sum_{k=1}^{m} \Psi_k$ must be at least $2\pi$, which clearly violates \eqref{eq:matrix_segmental_phase}. Consequently, for examining whether \eqref{eq:matrix_segmental_phase} is satisfied, it suffices to make a selection of the segmental phase of one matrix being at most, since other segmental phases are all singletons. 
\subsection{Comparison with Existing Phase Definitions}
The purpose of this subsection is to compare the segmental phase to some existing phase notions. These include the recent sectorial phase \cite{Wang:20} and the aged principal phase \cite{Postlethwaite:81}.

A matrix $A\in \cnn$ is said to be \emph{sectorial} 
if $\mathcal{N}(A)$ is contained in a minor segment. For a sectorial matrix $A$, there are two unique supporting rays of $\mathcal{N}(A)$ to form a \emph{smallest convex circular sector} anchored at the origin, and thus the opening angle of this sector is less than $\pi$. Then, the \emph{sectorial phase} of a sectorial matrix $A$ is uniquely defined by
 \begin{align*}
 \textstyle \Phi(A)\coloneqq\interval{\underline{\phi}(A)}{\overline{\phi}(A)} = \interval[scaled]{ \inf_{z \in \mathcal{N}(A)}\angle z}{ \sup_{z \in \mathcal{N}(A)} \angle z}.
\end{align*}
See a graphical illustration of the sectorial phase in Fig.~\ref{fig:seg_phase_vs_sec_phase}. The two terms ``segmental phase'' and ``sectorial phase'' signify their first and biggest difference vividly. The former exploits the smallest segment to bound $\mathcal{N}(A)$, while the latter adopts the smallest convex sector. Besides, the two notions differ in a few other significant aspects. The segmental phase is defined for all matrices, but the sectorial phase is subject to {the class of sectorial matrices} (slightly extensible to semi-sectorial matrices, e.g., \cite{Chen:21}). It is worth noting that the class of matrices in Proposition~\ref{prop: sectorial} is exactly the sectorial class. Hence the segmental phase of a sectorial matrix is unique.

\begin{figure}[htb]
 \vspace{-2mm}
 \centering
 \includegraphics[width=2.5in, trim={1.8cm 1.7cm 1.5cm 1.2cm}, clip]{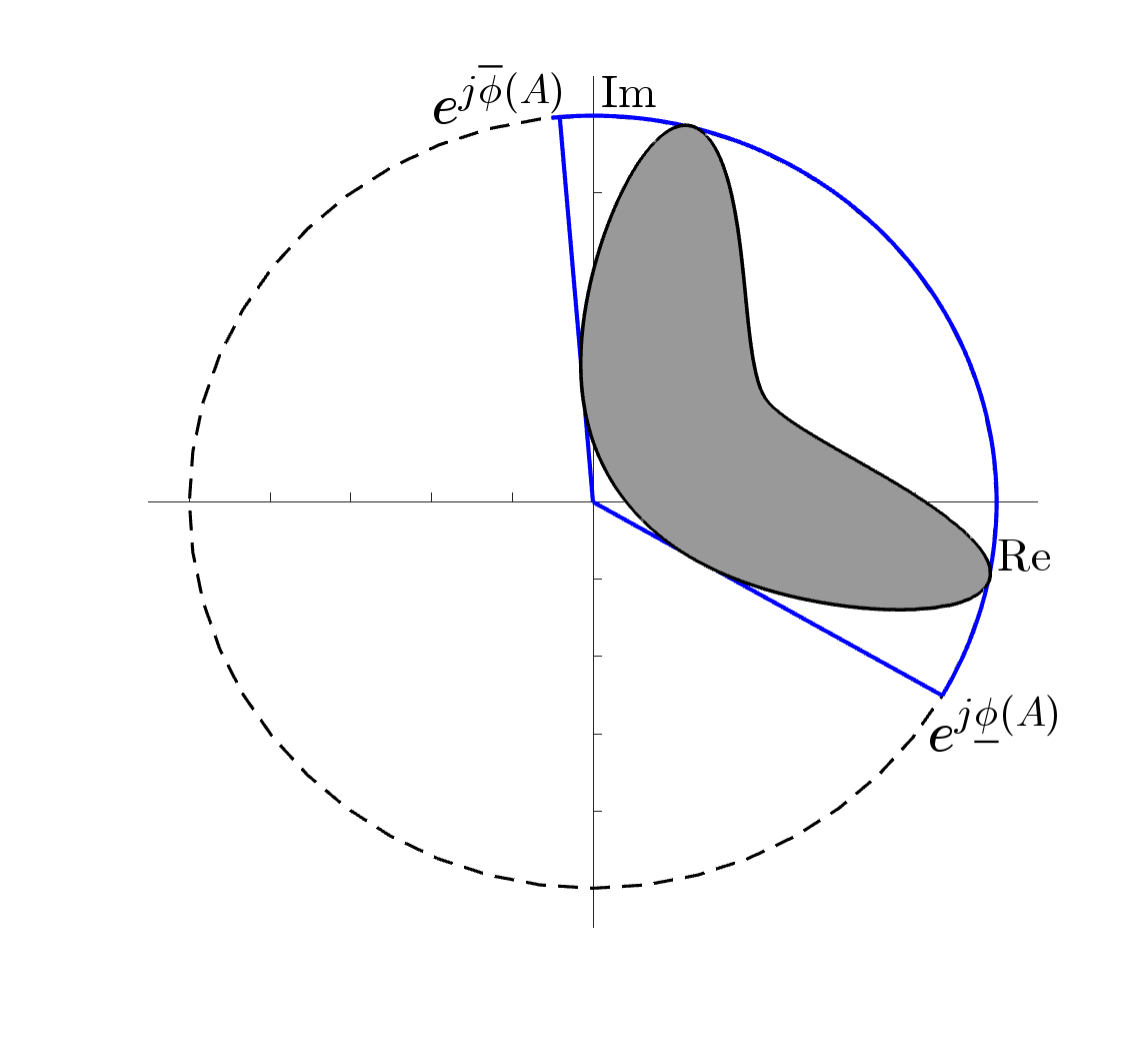}
 \vspace{-2mm}
 \caption{An illustration of the sectorial phase of a sectorial matrix $A$. The blue-bordered smallest convex sector is adopted to bound $\mathcal{N}(A)$. The two endpoints of the sector given by $e^{j\underline{\phi}(A)}$ and $e^{j\overline{\phi}(A)}$ define the sectorial phase $\Phi(A)=\interval{\underline{\phi}(A)}{\overline{\phi}(A)}$.}\label{fig:seg_phase_vs_sec_phase} 
\end{figure}

Much earlier in the history, the principal phase of a matrix was defined to be the segmental phase of the unitary part of its polar decomposition \cite{Postlethwaite:81}. To be precise, let nonsingular $A \in \cnn$ have the unique polar decomposition $A=UP$ with $U$ unitary and $P$ positive definite. Then the \emph{principal phase} is defined to be $\Phi_{p}(A)\coloneqq \Psi (U)$. Possibly this is multi-interval-valued. When $A$ is sectorial, then $U$ is also sectorial. In this case $\Phi_p(A)$ is a singleton. 

By the example in Fig.~\ref{fig:seg_phase_vs_sec_phase}, one can easily imagine that taking the smallest segment will yield a larger phase interval than the one obtained from the smallest sector. One may naturally wonder, given any sectorial matrix, whether its sectorial phase is always smaller than its segmental phase. In addition, what can we say about the principal phase? The following proposition gives an order relationship among the three:

\begin{proposition}\label{prop: sec_seg}
For sectorial $A\in \cnn$, it holds that $$\Phi_p(A) \subset \Phi(A) \subset \Psi(A).$$
\end{proposition}
\begin{proof}
The first containment is proved in \cite[p. 159]{Wang:20}. To show $\Phi(A) \subset \Psi(A)$, let $\Psi(A)=\interval{\alpha}{\beta}$ be the segmental phase interval which corresponds to the two endpoints $e^{j\alpha}$ and $e^{j\beta}$ of the arc of the smallest segment. Note that $\beta-\alpha<\pi$. By picking the two points $e^{j\alpha}$ and $e^{j\beta}$, construct a convex sector $\tilde{\mathcal{C}}$ whose phase interval from the same arc is also $\interval{\alpha}{\beta}$ and $\mathcal{N}(A)\subset \tilde{\mathcal{C}}$. Since the sectorial phase interval $\Phi(A)$ comes from the smallest convex sector $\mathcal{C}$ covering $\mathcal{N}(A)$, we conclude that $\mathcal{C}\subset \tilde{\mathcal{C}}$ and $\Phi(A) \subset \interval{\alpha}{\beta} = \Psi(A)$.
\end{proof}

This proposition somewhat suggests that for sectorial matrices only, the principal phase and sectorial phase seem like better choices. However, we highlight that \emph{there is no free lunch in making a phase interval smaller}. The benefit of the ``smaller'' principal phase or the sectorial phase comes with a price when considering the product of multiple (more than two) matrices in terms of its product eigen-phase bound. The bound for the principal phase only holds for a single matrix, i.e., we only have $\angle \lambda_i (A) \in \Phi_p (A)$\cite[Th.~2]{Postlethwaite:81}. For the sectorial phase the eigen-phase bound holds only for the product of two matrices, i.e., $\angle \lambda_i (A B) \in \Phi (A) + \Phi (B)$ \cite[Th.~6.2]{Wang:20}. For the segmental phase, the product eigen-phase bound holds for the product of an arbitrary number of matrices (Theorem~\ref{lem:eigen_phase}). This shows that only the segmental phase suits well in studying a cyclic structure like Theorem~\ref{thm:matrix_segmental_phase}.

We conclude this subsection by distilling the usage of the newly-established segmental phase in contrast to that of the sectorial phase. First, the segmental phase is perfectly suited for Problem~\ref{problem_matrix} with multiple possibly uncertain matrices. The reason why the sectorial phase is inapplicable is straightforward: Problem~\ref{problem_matrix} cannot be converted and grouped into a two-matrix problem without introducing an extra interconnection matrix and the ordering of the matrices is generally unchangeable. Second, the segmental phase works for all matrices which gets rid of the sectorial-class restriction in the sectorial phase. 

\section{The Segmental Phase of A MIMO LTI System}\label{sec:systems}

Let us begin with the celebrated notion of $\hinf$-norm \cite[Sec.~4.3]{Zhou:96} of a stable MIMO system $P\in \rhinf^{n\times n}$:
\bex\label{eq:system_gain}
\norm{P}_\infty\coloneqq \textstyle \sup_{\omega\in \interval{0}{\infty}} \overline{\sigma}(P(\jw)).\eex
Here, $\overline{\sigma}(P(\jw))$, a continuous function of frequency $\omega$, is known as the system magnitude or \emph{gain response}. As a counterpart to the gain, in this section we develop the \emph{phase response} for a MIMO system based on the matrix segmental phase. Unlike the gain definition restricted for stable systems only, the phase definition below can apply to \emph{semi-stable systems} which admit $\jw$-axis poles. Additional notation on indented contours is required before we move on. 
\subsection{The Indented Nyquist Contour and Preliminaries}
Consider a semi-stable system $P\in \mathcal{R}^{n\times n}$ with full normal rank. Denote by ${\Omega}^p$ ($\Omega^z$, resp.) the following set of nonnegative frequencies from the $\jw$-axis poles (zeros, resp.) of $P(s)$: $\Omega^p \coloneqq \bbkt{\omega\in \interval[open right]{0}{\infty}\mid s=\jw~\text{is a pole of}~P(s) }$ and ${\Omega}^z \coloneqq \bbkt{\omega\in \interval{0}{\infty}\mid s=\jw~\text{is a zero of}~P(s)}$. The notation $j\Omega\coloneqq \bbkt{\jw \mid \omega \in \Omega}$ and $j\Omega_{\pm} \coloneqq \bbkt{\pm \jw \mid \omega \in \Omega}$ can also be adopted unambiguously, where
\be\label{eq: pole_zero_set}
\Omega\coloneqq \Omega^p \cup \Omega^z.
\ee Given $\omega_0\in \mathbb{R}$ and sufficiently small $\epsilon> 0$, let $\mathbf{SC}(\epsilon, j\omega_0)$ be the semi-circle in $\ccp$ with the center $j\omega_0$ and radius $\epsilon$, namely, 
$\mathbf{SC}(\epsilon, \jw_0)\coloneqq \bbkt{s\in \mathbb{C} \mid \abs{s-j\omega_0}=\epsilon, \rep\rbkt{s}\geq 0}$.
Based on $\mathbf{SC}(\epsilon, j\omega_0)$, define the {indented Nyquist contour} $\mathbf{NC}$ as depicted in Fig.~\ref{fig:nyquist}. The contour $\mathbf{NC}$ has semicircular indentations with radius $\epsilon$ around $\pm \jw_0$, where $\omega_0\in {\Omega}$. In the case of $\omega_0 = \infty$, a semicircular indentation with radius $1/\epsilon$ is taken. 
 Define the \emph{leading coefficient matrix of $P(s)$ at a pole $j\omega_0\in j\Omega^p$} of order $l$ by 
 \begin{align}\label{eq:leading_coeff_pole}
 {K^p(\omega_0)}\coloneqq \lim_{s\to j\omega_0} {(s-j\omega_0)^{l}}P(s).
 \end{align}
 In parallel, define the \emph{leading coefficient matrix of $P(s)$ at a zero $j\omega_0\in j\Omega^z$} of order $l$ by 
\begin{align}\label{eq:leading_coeff_zero}
 {K^z(\omega_0)}\coloneqq \lim_{s\to j\omega_0} {(s-j\omega_0)^{-l}}P(s).
 \end{align}
 Throughout this paper, we consistently deal with systems satisfying the following technical assumption:
\begin{assumption}\label{assum: system}
 Suppose that $P\in \mathcal{R}^{n\times n}$ is semi-stable with full normal rank, $K^p(\omega_0)$ in \eqref{eq:leading_coeff_pole} has full rank for all $\omega_0\in \Omega^p$, and $K^z(\omega_0)$ in \eqref{eq:leading_coeff_zero} has full rank for all $\omega_0\in \Omega^z$.
\end{assumption}

The reason for adding Assumption~\ref{assum: system} will be made clear soon. Note that it precludes those systems having $\jw$-axis poles and zeros at the same locations. This is consistent with SISO semi-stable systems $\mathcal{R}^{1\times 1}$.

\begin{figure}[htb]
 \vspace{-1mm}
 \centering
 \setlength{\unitlength}{0.8mm}
 \begin{picture}(70,50)
 \thicklines
 \put(10,25){\vector(1,0){50}}
 \put(30,0){\vector(0,1){50}}
 {\thicklines
 \put(34,22){\makebox(0,0){$0$}}
 \thicklines
 {\thinlines \put(30,25){\vector(3,2){16.7}}}
 \put(30,25){\vector(0,1){5}}
 \put(30,38){\vector(0,1){5}}
 \put(34,31){\makebox(0,0){$\frac{1}{\epsilon}$}}
\put(30, 25){\arc[0,90]{20}}
 \put(30, 25){\arc[-90,0]{20}}
 \put(32,35){\makebox(0,0){\scriptsize$\epsilon$}}
 {\thinlines \tiny \put(30,35){\vector(3,2){2.7}}}
 \put(30, 35){\arc[0,90]{3}}
 \put(30,35){\makebox(0,0){$\times$}}
 \put(30, 35){\arc[-90,0]{3}}
 \put(30, 15){\arc[0,90]{3}}
 \put(30,18){\vector(0,1){5}}
 {\thinlines \tiny \put(30,15){\vector(3,2){2.7}}}
 \put(30, 15){\arc[-90,0]{3}}
 \put(30,5){\vector(0,1){5}}
 \put(32,15){\makebox(0,0){\scriptsize$\epsilon$}}
 \put(30,15){\makebox(0,0){$\times$}}
 \put(59,22){\makebox(0,0){$\rep$}}
 \put(35,48){\makebox(0,0){$\imp$}}}
 \end{picture} \caption{An indented contour $\mathbf{NC}$ in $\ccp$, where ``$\times$'' denotes $j\omega$-axis poles or zeros and $\epsilon$ is taken to be sufficiently small. } \label{fig:nyquist}
 \end{figure}
 
 Given $P$ satisfying Assumption~\ref{assum: system}, we will mainly exploit the frequency-response matrix $P(\jw)$ to define the segmental phase of $P(s)$. When $P(s)$ has a pole at $s=j\omega_0$, an obvious difficulty arises for only considering $P(\jw)$, since $P(\jw_0)$ is not even well defined. A reasonable ``phase value'' of $P(\jw_0)$ is supposed to be \emph{empty} following from the convention of $\angle \infty=\emptyset$. To handle the difficulty, we need the indented contour in Fig.~\ref{fig:nyquist} for determining \emph{possible phase changes} of $P(\jw)$ around the $\jw$-axis poles. The same difficulty and treatment even exist in the SISO case when plotting the Bode phase diagram, e.g., $\frac{1}{s^2+1}$ at $j\omega_0=j1$. The principle behind the treatment is more or less standard: {As $\omega$ increases and $s$ travels around $\semicir$ counterclockwise, a certain ``phase value'' of $P(\jw)$ should decrease by $l\pi$, where $j\omega_0$ is a pole of order $l$}. For this reason, there is no ambiguity to determine ``phase values'' around $P(\jw)$ in the sense that {the contour is implicitly taken by convention}. We refer the reader to \cite[p.~93]{Desoer:75} for analogous consideration as above. A similar story occurs around a $\jw$-axis zero of $P(s)$ in which a zero-indented-contour is inevitably required for determining appropriate ``phase values''. For example, the case $\frac{s^2+1}{(s+1)^2}$ at $j\omega_0=j1$. Intuitively, a certain ``phase value'' of $P(\jw)$ should increase by $\pi$ times the order of the zero $j\omega_0$, and a ``phase value'' of $P(\jw_0)$ should be assigned to be \emph{empty} similarly to the convention of $\angle 0 = \emptyset$.
 
\subsection{The Segmental Phase Response of Systems}\label{sec:system_segphase_def}
Consider a semi-stable system $P\in \mathcal{R}^\nn$ with the set of $\jw$-axis poles and zeros $\Omega$ in \eqref{eq: pole_zero_set}. We visualize the normalized numerical range $\mathcal{N}(P(\jw))$, a frequency-dependent set continuous in $\omega\in \interval{0}{\infty}\setminus \Omega$ contained in the unit disk. Similarly to the matrix case in Fig.~\ref{fig:seg_phase}, graphically we tailor a frequency-dependent smallest circular segment to cover $\mathcal{N}(P(\jw))$ for each $\omega\in \interval{0}{\infty}\setminus\Omega$. This immediately defines
the \textit{(frequency-wise)} \textit{segmental phase} of $P(s)$:
\be\label{eq:system_phase}
 \Psi(P(\jw))\coloneqq \interval[scaled]{ \underline{\psi}(P(\jw))}{\overline{\psi}(P(\jw))},
\ee
 where $\omega\in \interval{0}{\infty}\setminus \Omega$, and $\Psi(P(\jw))=\emptyset$ whenever $\omega\in \Omega$. The \textit{(frequency-wise)} \textit{phase center} 
 \be\label{eq:system_center}
 \gamma^\star(P(\jw))\coloneqq \textstyle\frac{1}{2}\rbkt{\underline{\psi}(P(\jw))+\overline{\psi}(P(\jw))}\in \mathbb{R}
 \ee can be defined analogously. Following the matrix phase in Definition~\ref{def: segphase_matrix}, a multi-interval-valued issue may arise when specifying $\Psi(P(\jw))$ for each $\omega$. A further assumption on a class of systems under consideration is made below that renders the segmental phase $\Psi(P(\jw))$ \emph{frequency-wise unique}.

 \begin{assumption}\label{assum: singleton}
Suppose that the phase center $\gamma^\star(P(\jw))$ defined in \eqref{eq:system_center} is a singleton (i.e. unique) for all $\omega\in \interval{0}{\infty}\setminus\Omega$.
 \end{assumption}

Given $P$, one can conduct a frequency-wise test to examine whether Assumption~\ref{assum: singleton} holds. The associated behind-the-scenes optimization problem \eqref{eq:phase_center_response} in Appendix~\ref{appendix:phase} is required to have a unique solution for each $\omega\in \interval{0}{\infty}\setminus\Omega$. This yields that $\gamma^\star(P(\jw))$ is continuous in $\omega\in \interval{0}{\infty}\setminus\Omega$ as a by-product. For the sake of technical simplicity, throughout we \emph{restrict} the segmental phase definition \eqref{eq:system_phase} to the following class of systems unless otherwise specified:
\be\label{eq: system_class}
\mathcal{P}^\nn\coloneqq \bbkt{P \mid \text{Assumptions~\ref{assum: system} and \ref{assum: singleton} are satisfied}}.
\ee
The class $\mathcal{P}^\nn$ is in fact \emph{generic} in the sense that the majority of matrices have unique segmental phases. Moreover, according to Proposition~\ref{prop: sectorial}, $\mathcal{P}^\nn$ contains the class of frequency-wise sectorial systems introduced in \cite{Chen:21}, \cite{Chen:20j}.

Definition \eqref{eq:system_phase} implicitly stands on the use of indentations around the $\jw$-axis poles/zeros. When $s$ moves along $\semicir$ for $\omega_0\in \Omega$, the full-rank conditions in Assumption~\ref{assum: system} guarantee that the phase change of $\Psi(P(s))$ approximates to scalar addition/subtraction to $\Psi(P(s))$. The reason is straightforward: for $s\in \semicir$, $P(s)$ is dominated by the matrix ${K^p(\omega_0)}$ in \eqref{eq:leading_coeff_pole} or ${K^z(\omega_0)}$ in \eqref{eq:leading_coeff_zero}. In the case when $0\in \Omega$, the segmental phase $\Psi(P(\jw))$ starts at $\omega=\epsilon$ with sufficiently small $\epsilon>0$ and hence we require the additional information $\Psi(P(\epsilon))$. Then $\Psi(P(j\epsilon))$ is uniquely determined from $\Psi(P(\epsilon))$ by taking a quarter-circle detour around $j0$.

Note that $\Psi(P(\jw))$ can be termed as the \emph{phase response} of $P$. An illustrative example of $\Psi(P(\jw))$ is provided below to facilitate our graphical understanding. In Fig.~\ref{fig:seg_phase_exmp}, the new \emph{Bode-type phase response} takes shape like a ``{river}'' whose two sides are exactly covered by $\overline{\psi}(P(\jw))$ and $\underline{\psi}(P(\jw))$. 

\begin{figure}[htb]
 \vspace{-3mm}
 \centering 
 \includegraphics[width=3.5in, trim={1cm 0.3cm 1.4cm 0.5cm}, clip]{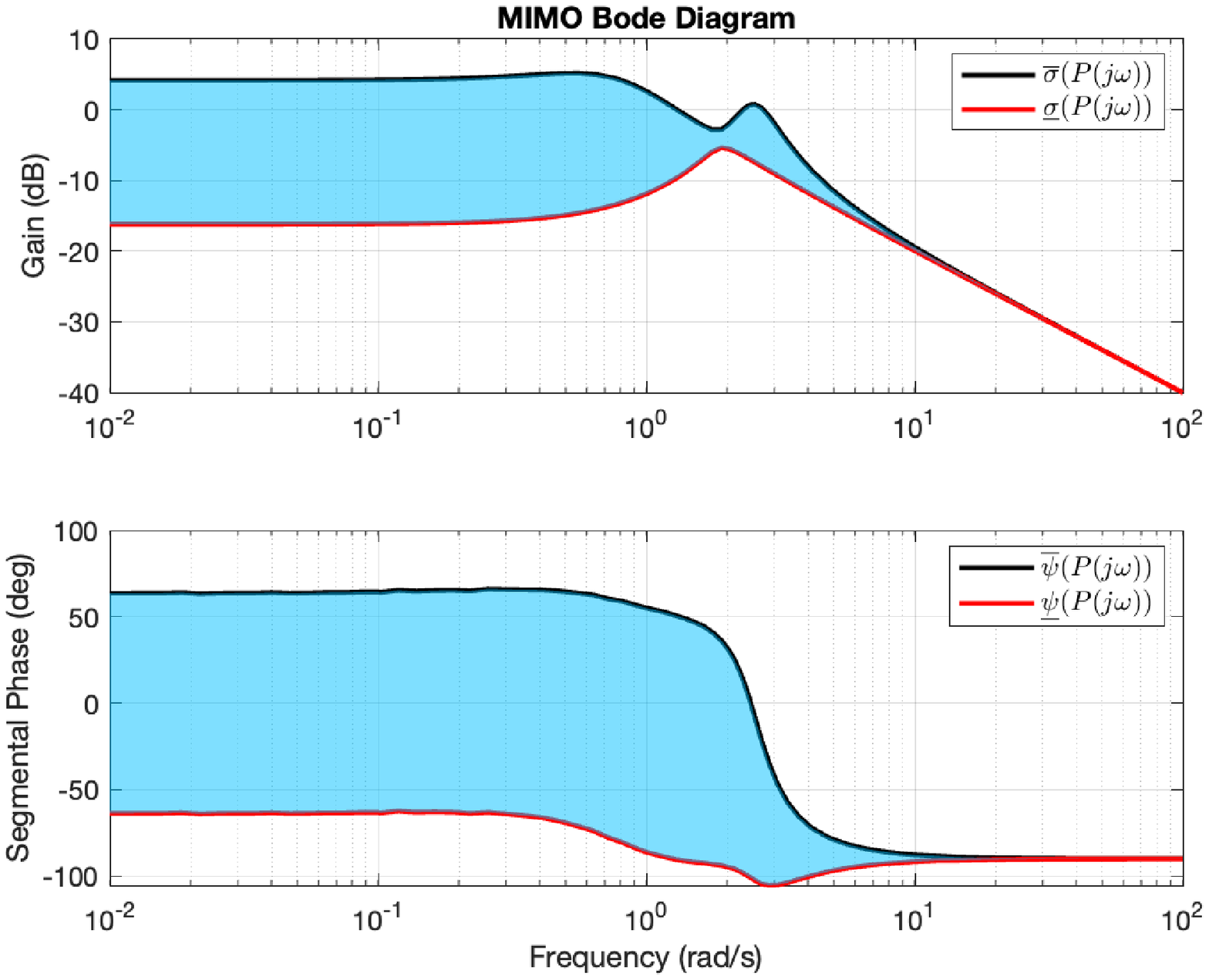} 
 \caption{A complete MIMO Bode diagram of $P(s)$ in \eqref{eq:example} over $\interval{0.01}{100}~\textrm{rad/s}$. Top: gain response $\interval{\underline{\sigma}(P(\jw))}{\overline{\sigma}(P(\jw))}$. Bottom: phase response $\interval[scaled]{ \underline{\psi}(P(\jw))}{\overline{\psi}(P(\jw))}$.}\label{fig:seg_phase_exmp}
\end{figure}

\begin{example}
Let $P$ be a matrix second-order mechanical system with a combined position-and-velocity output given by $P(s)=(1+s)(s^2M+sD+K)^{-1}$, where the mass matrix and the damping matrix are identity, $M=D=I$, and the stiffness matrix $K=\stbt{5}{3}{2}{2}$. Explicitly, we have 
 \be\label{eq:example}
 P(s)=\frac{\tbt{s^3+2s^2+3s+2}{-2s-2}{-3s-3}{s^3+2s^2+6s+5}}{s^4+2s^3+8s^2+7s+4}.
 \ee
The gain response and the newly defined phase response of $P(s)$ are depicted in Fig.~\ref{fig:seg_phase_exmp} side-by-side. It shares a similar flavor with the classical Bode gain-phase diagram. 
 \end{example}

In Fig.~\ref{fig:seg_phase_exmp}, the segmental phase of the system \eqref{eq:example} satisfies $\Psi(P(\jw))\in \interval[open right]{-\pi}{\pi}$ for all $\omega$; that is, the width of the ``river'' is uniformly less than $2\pi$. In this case, at every frequency $\omega$, the system phase response coincides with the phase interval of the corresponding matrix $A=P(\jw)$. However, for a general system $P\in \mathcal{P}^\nn$, its appropriate phase intervals may not always be the case \emph{limited to a $2\pi$-interval}, since the continuity of the frequency response $P(\jw)$ must be taken into consideration. Such an issue can be easily grasped from a SISO example $\frac{1}{(s+1)^6}$ whose phase response is continuously decreasing from $0$ to $-3\pi$.

In the MIMO case, as $\omega$ increases, determining an appropriate phase interval $\Psi(P(\jw))$ demands first specifying its appropriate phase center $\gamma^{\star}(P(\jw))$ in \eqref{eq:system_center}. The matrix phase center in Definition~\ref{def: segphase_matrix} is restricted in the principal branch $\interval[open right]{-\pi}{\pi}$ which is no longer suitable here. To remedy the issue algorithmically, when solving $\gamma^\star(P(\jw))$ in \eqref{eq:phase_center_response} starting from $\omega=0$ or $\epsilon$, we designate a principal branch $\interval[open right]{\alpha}{\alpha+2\pi}$ with a specific $\alpha\in \mathbb{R}$ such that $\gamma^\star(P(\jw))$ can be found within $\interval[open right]{\alpha}{\alpha+2\pi}$. As $\omega$ increases, we gradually shift $\alpha$ to adapt $\gamma^\star(P(\jw))$ for a new principal branch such that $\gamma^\star(P(\jw))$ is continuous for $\omega\in \interval{0}{\infty}\setminus\Omega$ within the branch. The uniqueness in Assumption~\ref{assum: singleton} is interpreted modulo $2\pi$ in terms of the principal branch. Such a treatment is in line with constructing Riemann surfaces \cite{Macfarlane:77} for appropriate values of the system's eigenloci.

To summarize, intuitively the phase response $\Psi(P(\jw))$ as a ``river'' continuously carry {``water''} without interruptions in every component of the ``river'' split by a finite number of pole/zero locations $\Omega$. Consequently, the jump discontinuities in $\Psi(P(\jw))$ only occur at $\omega\in \Omega$.

\begin{remark}[Back to the classical phase]\label{rem:classical_phase} For semi-stable SISO systems,
 the segmental phase \eqref{eq:system_phase} recovers to the classical phase notion. To see this, let $P\in \mathcal{R}^{1\times 1}$ be semi-stable and then Assumption~\ref{assum: system} holds trivially. Using Example~\ref{example2}\ref{item:prop_1}, straightforward calculation yields that
 $\Psi(P(\jw)) 
 ={\angle P(\jw)}$
 for all $\omega\in \interval{0}{\infty}$. More importantly, our technicalities used for addressing a well-behaved phase response, although looking complicated, can be easily grasped once the reader fully comprehends how to sketch Bode phase diagrams for simple systems like $\frac{1}{(s+1)^6}$, $\frac{1}{s^2}$, $\frac{s}{s^2+1}$ and $\frac{s^2+1}{s(s+1)^2}$. These examples demonstrate vividly the issues of appropriate phase responses regarding principal branches and poles/zeros.
 \end{remark} 
 
\subsection{A Multi-Interval-Valued Case}\label{sec: multi-valued}
The purpose of this subsection is to define the segmental phase for a special class of systems as a \emph{supplement} to the class $\mathcal{P}^\nn$ in \eqref{eq: system_class}. For this particular class, the system segmental phase is \emph{frequency-wise multi-interval-valued}. Nevertheless, this multi-valued issue can be clearly addressed based on the matrix phase. 

Consider a system $P\in \mathcal{R}^\nn$ that can be decomposed into
\be\label{eq: system_multi_value}
P(s)=Q(s)A,
\ee where $Q\in \mathcal{R}^{1\times 1}$ is a semi-stable SISO system and $A\in \mathbb{R}^\nn$ satisfies that its segmental phase $\Psi(A)$ is multi-interval-valued. The set of the systems satisfying \eqref{eq: system_multi_value} is denoted by $\mathcal{Q}^\nn$. A commonly-seen MIMO integrator $\frac{1}{s}\stbt{0}{1}{1}{0} \in \mathcal{Q}^{2\times 2}$ provides such an example since $\Psi(\stbt{0}{1}{1}{0})=\{\interval{0}{\pi}, \interval{-\pi}{0}\}$. The segmental phase developed in Section~\ref{sec:system_segphase_def} can be handily adapted to this particular class. For $P\in \mathcal{Q}^\nn$, the normalized numerical range $\mathcal{N}(P(\jw))$ can be constructed from $\mathcal{N}(A)$ scaled by a frequency-wise scalar $Q(\jw)\in \mathbb{C}$. In light of \eqref{eq: rotation_NNR}, we obtain $\mathcal{N}(P(\jw)) = e^{j\angle Q(\jw)}\mathcal{N}(A)$. Thus, the segmental phase of $P\in \mathcal{Q}^\nn$ can be defined as
\be\label{eq: system_multi_valued_phase}
\Psi(P(\jw)) = \bbkt{ \angle Q(\jw)+\Psi \mid \Psi \in \Psi(A)},
\ee 
where $\omega\in \interval{0}{\infty}\setminus\Omega$ and $\Psi(P(\jw))=\emptyset$ when $\omega\in \Omega$. $\Psi(P(\jw))$ is frequency-wise multi-interval-valued fully inherited from $\Psi(A)$. Intuitively, $\Psi(P(\jw))$ contains $N$ continuous phase responses which shape like $N$ ``rivers'' of equal width, where $N$ denotes the cardinality of $\Psi(A)$. We next show an example to facilitate the understanding of $N$ phase responses.

\begin{example}\label{exmaple3}
 Consider $P(s)= \frac{1}{s^2+1}A$, where $A=\stbt{0}{1}{-1}{0}$. Note that $\mathcal{N}(A)=\{z\mid \rep (z) =0, -1\leq \imp(z) \leq 1 \}$, the line segment connecting $-j$ to $j$, and $\Psi(A)=\bbkt{\interval{-\frac{\pi}{2}}{\frac{\pi}{2}}, \interval{\frac{\pi}{2}}{\frac{3\pi}{2}}}$ is two-interval-valued with its center $\gamma^\star (A)=\{0, \pi\}$. Hence we have two phase responses. For all $\omega\in \interval[open right]{0}{1} \cup \interval[open]{1}{\infty}$, note that $\mathcal{N}(P(\jw))$ keeps the same region as $\mathcal{N}(A)$ since $\angle \frac{1}{1-\omega^2}$ is either $0$ or $\pi$. One should recognize that a $\pi$-phase-shift occurs across the $j1$-pole by analyzing $\mathcal{N}(P(s))$ for $s\in \mathbf{SC}(\epsilon, j1)$. By~\eqref{eq: system_multi_valued_phase}, we obtain the first phase response ${\Psi}(\jw)\in \Psi(P(\jw))$: 
 \begin{equation*}
 \Psi(\jw)=\left\{ 
\begin{aligned} 
 & \textstyle \interval[scaled]{-\frac{\pi}{2}}{\frac{\pi}{2}}\quad \quad \forall \omega \in \interval[open right]{0}{1},\\
 & \emptyset \hspace{0.6in}\;\; \text{when}~\omega=1, \\
& \textstyle \interval[scaled]{-\frac{3\pi}{2}}{-\frac{\pi}{2}}\quad \forall \omega\in \interval[open]{1}{\infty}
\end{aligned}
 \right.
\end{equation*}
 and the second one: $\Psi(\jw)=\interval[scaled]{\frac{\pi}{2}}{\frac{3\pi}{2}}$ when $\omega \in \interval[open right]{0}{1}$; $\Psi(\jw)=\emptyset$ when $\omega=1$; $\Psi(\jw)=\interval[scaled]{-\frac{\pi}{2}}{\frac{\pi}{2}}$ when $\omega\in \interval[open]{1}{\infty}$.
 \end{example}
 
From now on, we bring our focus back to the system segmental phase for the general class $\mathcal{P}^\nn$ in \eqref{eq: system_class}, and thereby the notation of unique $\Psi(P(\jw))$ is unambiguous in this context. The segmental phase for the special class $\mathcal{Q}^\nn$ will only be included in several remarks as additional results to our main results established below.
 
\section{Main Results}\label{sec:small_phase}

This section presents the main result of this paper, a \emph{cyclic small phase theorem}, as a fundamental solution to Problem~\ref{problem}. The theorem states a brand-new stability condition involving the ``loop system phase'' being contained in $\interval[open]{-\pi}{\pi}$, which complements the famous small gain theorem. Furthermore, a mixed small gain/phase theorem is established by intertwisting system gain/phase information in a frequency-wise manner.
 
Consider a cyclic feedback system, where $P_k\in \mathcal{P}^\nn$ with the $j\omega$-axis pole/zero set $\Omega_k$ in \eqref{eq: pole_zero_set} for $k=1, 2, \ldots, m$. Denote by $\Omega^p\coloneqq \bigcup_{k=1}^{m} \Omega_k^p$, $\Omega^z\coloneqq \bigcup_{k=1}^{m} \Omega_k^z$ and $\Omega\coloneqq \bigcup_{k=1}^m \Omega_k$. Since the cyclic loop is assumed to have {no unstable pole-zero cancellation}, by Lemma~\ref{lem:feedback_no_cancellation}, if some subsystem has a pole (a zero, resp.) at $j\omega_0$, then the remaining subsystems cannot have zeros (poles, resp.) at $j\omega_0$, that is, $\Omega^p \cap \Omega^z =\emptyset$.
 
\subsection{A Cyclic Small Phase Theorem}

Recall a simpler version of the small gain condition \eqref{eq:small_gain_cyc} for the cyclic feedback stability for $P_k\in \rhinf^{n\times n}$, that is,
\be\label{eq:small_gain_simple}
 \prod_{k=1}^m\overline{\sigma}(P_k(\jw)) <1,
 \ee
where $\omega\in\interval{0}{\infty}$. We are now ready to state the main result which stands side-by-side with the small gain theorem. The result lays the foundation of the segmental phase study. 
\begin{theorem}\label{thm:segmental_phase_stability}
 For $P_1, P_2, \ldots, P_m \in \mathcal{P}^\nn$ in \eqref{eq: system_class}, the cyclic feedback system is stable if 
 \be\label{eq:thm:seg_phase_stability}
 { \sum_{k=1}^{m} \Psi(P_k(\jw))} \subset \interval[open]{-\pi}{\pi},
 \ee
where $\omega\in \interval{0}{\infty}$.
\end{theorem}
\begin{proof}
See Appendix~\ref{sec:proofs}.
\end{proof}

In the case of
$P_k\in \rhinf^{n\times n}$ with $\Omega_k=\emptyset$ for $k=1, 2, \ldots, m$, Theorem~\ref{thm:segmental_phase_stability} reduces to~\cite[Th.~2]{Chen:22IFAC} in our conference version. By Remark~\ref{rem:classical_phase}, for SISO semi-stable systems $P_k\in \mathcal{R}^{1\times 1}$, the \emph{small phase condition} \eqref{eq:thm:seg_phase_stability} recovers to the classical condition: $\sum_{k=1}^m \angle P_k(\jw) \in \interval[open]{-\pi}{\pi}$. 
Let us now reveal the core messages enclosed in Theorem~\ref{thm:segmental_phase_stability}. 

Condition~\eqref{eq:thm:seg_phase_stability} involves a comparison of the sum of the phases (or the loop-phase interval) to the interval $\interval[open]{-\pi}{\pi}$ in size. It is worthy of putting \eqref{eq:small_gain_simple} and \eqref{eq:thm:seg_phase_stability} into the context of the generalized Nyquist criterion \cite{Desoer:80} which underpins the proof of Theorem~\ref{thm:segmental_phase_stability}. The criterion counts the number of encirclements of the critical point ``$-1$'' made by the closed paths formed by the eigenloci of $P_m(s)P_{m-1}(s)\cdots P_1(s)$ when $s$ travels up the Nyquist contour. In the small gain case \eqref{eq:small_gain_simple}, the eigenloci are restricted in a simply connected region -- the unit disk, so that there is no encirclement of ``$-1$'' and thereby the stability is guaranteed regardless of the arguments of the eigenloci. The small phase case \eqref{eq:thm:seg_phase_stability} is developed in a parallel manner thanks to the vital product eigen-phase bound in Theorem~\ref{lem:eigen_phase}. The eigenloci now are restricted in another \emph{simply connected region -- the open $\interval[open]{-\pi}{\pi}$-cone} so that the eigenloci never cross the negative real axis and thereby never encircle ``$-1$''. Hence the stability is concluded \emph{regardless of the magnitudes of the eigenloci}. Such an idea is graphical.

\begin{remark}
 The important idea of classical phase lead-lag compensation is naturally rooted in \eqref{eq:thm:seg_phase_stability} which allows a phasic trade-off among all subsystems $P_k(s)$. In practice, some of subsystems may have large phase lags over certain frequency ranges. This can be compensated by the remaining subsystems with adequate phase leads over the same frequency ranges. Theorem~\ref{thm:segmental_phase_stability} thus enables a loop-phase-shaping design problem via the segmental phase and we leave it for future research. 
\end{remark}

\begin{remark}
 Condition~\eqref{eq:thm:seg_phase_stability} suggests a robustness phase-indicator of a cyclic loop. Specifically, the following quantity:
\begin{equation*} 
 \hspace{-1.2mm}\min_{\omega \in \interval{0}{\infty}}\; \textstyle\pi - \max\bbkt{ \abs{\sum_{k=1}^m \overline{\psi}(P_k(\jw))}, \abs{\sum_{k=1}^m \underline{\psi}(P_k(\jw))}}
\end{equation*}
characterizes the worst case of ``phase margin'' of a cyclic loop over all positive constant scalings in the loop. 
\end{remark}

Two hidden technical ingredients of Theorem~\ref{thm:segmental_phase_stability} need to be highlighted. Firstly, condition \eqref{eq:thm:seg_phase_stability} is only nontrivial for $\omega\in \interval{0}{\infty}\setminus \Omega$. For each $j\omega_0 \in j\Omega$ of order $l$, we indent a semi-circle $\semicir$. The proof of Theorem~\ref{thm:segmental_phase_stability} claims the following: \eqref{eq:thm:seg_phase_stability} also holds for $s \in \semicir$. Assumption~\ref{assum: singleton} plays a key role in the claim, guaranteeing that {$\sum_{k=1}^{m} \Psi(P_k(s))$} entirely decreases or increases by $l\pi$ when $s$ travels along $\semicir$. The reason behind it is rather simple: Except for $l\pi$, there is no other type of phase-value jumps around $s \in \semicir$ since $\Psi(P_k(s))$ is unique for all $k$. 

Secondly, condition~\eqref{eq:thm:seg_phase_stability} together with Assumption~\ref{assum: singleton} suggests that any $j\omega_0\in j\Omega$ must be \emph{at most of order $2$}. This is implicit but intuitive since $j\omega_0$ can contribute a total of $l\pi$-phase-shift. If $l>2$, clearly \eqref{eq:thm:seg_phase_stability} does not hold. One may feel puzzled about $l=2$ as the phase-shift seems to be already $2\pi$. In this case, the phase-shift may be an ``open'' rather than ``closed'' $2\pi$-shift owing to potential lead-lag compensation between subsystems. An illustrative example is the feedback of a phase-lag $P_1(s)=\frac{1}{s^2}$ and a phase-lead $P_2(s)=\frac{s+1}{s+2}$.
 Around $j0$, taking $s=\epsilon e^{j\alpha}$ gives $P_1(s)= \epsilon^{-2}e^{-j2\alpha}$ and $P_2(s)= \frac{1+\epsilon e^{j\alpha}}{2+\epsilon e^{j\alpha}}$. When $\alpha=\frac{\pi}{2}$, $\angle P_1(j\epsilon)=-\pi$. However, one can observe that $P_2(j\epsilon)=\frac{1+j\epsilon}{2+j\epsilon}$ provides a small phase-lead so that $\angle P_1(j\epsilon) +\angle P_2(j\epsilon) >-\pi$ still holds. 

Theorem~\ref{thm:segmental_phase_stability} can be expanded for further incorporating systems \eqref{eq: system_multi_value} in the special class $\mathcal{Q}^\nn$. In this circumstance, the segmental phase $\Psi(P(\jw))$ may have $N>1$ phase responses. In the following corollary, for $P_k\in \rbkt{\mathcal{P}^\nn\cup \mathcal{Q}^\nn}$, for brevity we adopt the unified notation $\Psi_k(\jw)\in \Psi(P_k(\jw))$ to indicate any selected phase response. 
\begin{corollary} 
For $P_1, P_2, \ldots, P_m \in \rbkt{\mathcal{P}^\nn\cup \mathcal{Q}^\nn}$, the cyclic feedback system is stable if there exist phase responses $\Psi_k(\jw) \in \Psi(P_k(\jw))$ with $k=1, 2, \ldots, m$ such that
 \bex
 \sum_{k=1}^{m} \Psi_k(\jw) \subset \interval[open]{-\pi}{\pi}, 
 \eex 
 where $\omega\in \interval{0}{\infty}$.
 \end{corollary}
 \begin{proof}
 The proof is omitted since it follows easily from combining Theorem~\ref{thm:matrix_segmental_phase} and the proof of Theorem~\ref{thm:segmental_phase_stability}.
 \end{proof}

It is worth noting that our segmental phase study can be adapted for discrete-time MIMO systems $P(z)$ by virtue of the matrix foundation laid
in Section~\ref{sec:matrix}. To be specific, one can adopt the frequency response matrix $P(e^{\jw})$ and define its segmental phase to be $\Psi(P(e^{\jw}))=\interval[scaled]{ \underline{\psi}(P(e^{\jw}))}{\overline{\psi}(P(e^{\jw}))}$,
where $\omega\in \interval{0}{\pi}$, in the same fashion as \eqref{eq:system_phase}. A discrete-time counterpart to condition~\eqref{eq:thm:seg_phase_stability} can be stated accordingly: $\sum_{k=1}^m \Psi(P_k(e^{\jw})) \subset \interval[open]{-\pi}{\pi}$.

\subsection{A Mixed Small Gain/Phase Theorem}
In many applications, using gain or phase analysis alone may not meet our practical needs. Very often, we should combine both gain and phase information \cite{Postlethwaite:81, Griggs:07, Forbes:10, Patra:11, Liu:15,Lestas:12, Zhao:22, Chaffey:21j, Chaffey:22_rolled, Woolcock:23}. A practical system usually has phase-lag increasing with input frequency and the phase shift for high-frequency input is very much unknown and even undefined due to noise. In particular, a large number of control loops have components with large (or infinity) gains
in the low frequencies and large (or undefined) phase lags in the high frequencies. 
The stability of such a loop can be established through a cut-off frequency by applying the gain condition \eqref{eq:small_gain_simple} to the high frequencies and the phase condition \eqref{eq:thm:seg_phase_stability} to the low frequencies. This simple idea can be further extended to a frequency-wise mixed gain/phase version: \emph{At each frequency, a gain condition or a phase condition is applied}. It is easy to digest the extension from a Nyquist perspective \cite{Desoer:80}: The eigenloci are contained in the set \emph{union} of the unit disk and $\interval[open]{-\pi}{\pi}$-cone -- a \emph{new simply connected region} away from ``$-1$''.

\begin{theorem}\label{thm:mixed_2}
 For $P_1, P_2, \ldots, P_m \in \mathcal{P}^\nn$ in \eqref{eq: system_class}, the cyclic feedback system is stable if for each $\omega\in \interval{0}{\infty}$, one of the following conditions holds: 
 \begin{enumerate}
 \renewcommand{\theenumi}{\textup{(\roman{enumi})}}\renewcommand{\labelenumi}{\theenumi}
 \item \label{item:mixedgain} $ \displaystyle \prod_{k=1}^m \overline{\sigma}(P_k(\jw))<1$;
 \item \label{item:mixedphase} { $\displaystyle\sum_{k=1}^{m} \Psi(P_k(\jw)) \subset \interval[open]{-\pi}{\pi}$.}
 \end{enumerate}
 \end{theorem}
 \begin{proof}
 The full proof is omitted for brevity since it follows similar lines of reasoning as in the proof of Theorem~\ref{thm:segmental_phase_stability}. The major difference lies in that in contrast to using the $\interval[open]{-\pi}{\pi}$-cone uniformly for all frequencies, instead we employ the union of the unit disk and $\interval[open]{-\pi}{\pi}$-cone defined as $\mathcal{D}\coloneqq \bbkt{z\in \mathbb{C}\mid \abs{z}<1~\text{or}~\angle z \in \interval[open]{-\pi}{\pi}}$. In such a case, all the eigenloci $\lambda_i(P(s))\in \mathcal{D}$ for all $s$ along the Nyquist contour $\mathbf{NC}$ and $i=1, 2, \ldots, n$, where $P= P_mP_{m-1}\cdots P_1$.
\end{proof}

 In Theorem~\ref{thm:mixed_2}, for those frequencies $\omega$ arbitrarily close to an element in $\Omega^p$, i.e., around a pole, the gain condition \ref{item:mixedgain} is impossible to be fulfilled and only the phase condition \ref{item:mixedphase} is possible to be applied to the frequencies. One notable specialization of Theorem~\ref{thm:mixed_2}, as mentioned earlier, is to intertwist gain/phase information from a given cut-off frequency $\omega_c\in \interval[open]{0}{\infty}$. Specifically, we require that $\sum_{k=1}^{m} \Psi(P_k(\jw)) \subset \interval[open]{-\pi}{\pi}$, where $\omega \in \interval[open right]{0}{\omega_c}$ and 
 $\prod_{k=1}^m \overline{\sigma}(P_k(\jw))<1$, where $\omega \in \interval{\omega_c}{\infty}$, similarly to that adopted in \cite[Th.~5]{Postlethwaite:81} for the principal phase and \cite[Th.~1]{Zhao:22} for the sectorial phase. To better demonstrate Theorem~\ref{thm:mixed_2}, we introduce the following example:
 \begin{example}
Let $P_1(s)=\frac{10}{s}I$, $P_2(s)=\stbt{\frac{1}{s+1}}{0.1}{0}{\frac{1}{s+1}}$ and $P_3= \stbt{\frac{s+1}{s+10}}{0}{0.1}{\frac{s+2}{s+5}}$. Note that $P_2(\jw)$ approaches to a nilpotent matrix as $\omega\to \infty$ and thereby its phase value becomes uninformative due to Example~\ref{example2}\ref{item:nilpotent}. Moreover, the gain of $P_1(\jw)$ is infinite at $\omega=0$ and it approaches to $0$ as $\omega\to \infty$, suggesting a gain condition for large $\omega$. By picking a cut-off frequency $\omega_c=2.6~\textrm{rad/s}$ and using Theorem~\ref{thm:mixed_2}, we can verify that the cyclic loop is stable since $\sum_{k=1}^{3} \Psi(P_k(\jw)) \subset \interval[open]{-\pi}{\pi}$ for all $\omega\in \interval{0}{\omega_c}$ and $\textstyle \prod_{k=1}^3 \overline{\sigma}(P_k(j\omega))<1$ for all $\omega\in \interval{\omega_c}{\infty}$. 
 \end{example}

The segmental phase is different from the \emph{quadratic} phase information in the Davis-Wielandt shell \cite{Lestas:12, Liang:24} known as another graphical tool. To the best of our knowledge, even for a two-subsystem feedback loop, both Theorems~\ref{thm:segmental_phase_stability} and~\ref{thm:mixed_2} cannot be recovered from \cite{Lestas:12}.

\section{Angular Scaling: The $\gamma$-Segmental Phase}\label{sec:generalized_small_phase} 
The purpose of this section is to introduce an angular scaling technique for stability analysis of cyclic feedback systems based on the segmental phase toward Problem~\ref{problem}. In a cyclic loop, considering that some of subsystems $P_k$ are known in advance, we propose to introduce tunable parameters $\gamma_k\in\interval[open right]{-\pi}{\pi}$ for each $P_k$ to make full use of the known information. Doing so will reduce the conservatism in Theorem~\ref{thm:segmental_phase_stability} by shaping (more precisely, reducing) the sum of the phases in the loop. Such a technique is largely inspired by the gain-scaling and $\mu$-analysis methods in gain-based control theory~\cite[Ch.~11]{Zhou:96}. The \emph{shaped phase} of a known matrix/system will be termed the \emph{$\gamma$-segmental phase}.

\subsection{The Matrix $\gamma$-Segmental Phase}
We begin with seeking an answer to Problem~\ref{problem_matrix}. Recall that $\mathcal{K}$ with $\mathcal{K}^\prime$ represents the index set of uncertain matrices/systems with its complement.
 The matrix segmental phase \eqref{eq:seg_phase_def} implicitly contains the phase center $\gamma^{\star}(A)$ which is optimal for a single matrix $A$. When invertibility of $I+ A_m A_{m-1}\cdots A_1$ is concerned in Problem~\ref{problem_matrix}, where $A_k$ or $\mathcal{N}(A_k)$ is known for $k\in \mathcal{K}^\prime$, each independent phase center $\gamma^{\star}(A_k)$ for $k\in \mathcal{K}^\prime$ however can bring conservatism in determining the invertibility. Our approach to reducing such conservatism is to allow a joint design of ``\emph{new centers}'' for those known matrices, resulting in a reduced sum of the phases in terms of \eqref{eq:matrix_segmental_phase}. Consider the following motivating example:
\begin{example}\label{example: gamma_phase}
Consider three ${2\times 2}$ matrices: $A_1= e^{j\frac{2\pi}{9}}D$ and $A_2=\mathrm{diag}(e^{j\frac{\pi}{6}}, e^{-j\frac{2\pi}{9}})$ are known and $A_3$ is uncertain in the sense of its segmental phase $\Psi(A_3)\subset\interval{-\frac{\pi}{2}}{\frac{\pi}{2}}$, where $D=\stbt{2}{-1}{-1}{2}$ is positive definite whose condition number $\kappa(D)=3$. According to \eqref{eq: rotation_NNR}, we have $\mathcal{N}(A_1) = e^{j\frac{2\pi}{9}}\mathcal{N}(D)$; namely, $\mathcal{N}(A_1)$ is a rotated line segment from $\mathcal{N}(D)$. It then follows from Example~\ref{example2}\ref{item:prop_3}-\ref{item:prop_4} that $\Psi(A_1)= \interval[scaled]{\frac{\pi}{18}}{\frac{7\pi}{18}}$ and $\Psi(A_2)= \interval[scaled]{-\frac{2\pi}{9}}{\frac{\pi}{6}}$. Since $A_3$ is uncertain, the worst case is 
 \bex
 \Psi(A_1) + \Psi(A_2) + \Psi(A_3)=\textstyle \interval[scaled]{-\frac{2\pi}{3}}{\frac{19\pi}{18}}
 \eex and thereby Theorem~\ref{thm:matrix_segmental_phase} is \emph{not} satisfied. However, in this simple example, one may group\footnote{In general, grouping and exchanging for uncertain matrices is impossible.}~$B\coloneqq A_2A_1$. It is known that the matrix $I+A_3B$ is invertible since $B$ is {strictly accretive} and $A_3$ is {accretive} \cite[Lem.~2.5]{Chen:21}. Note that the above conservatism can be reduced by designing \emph{a new segment} for $\mathcal{N}(A_1)$. For instance, to cover $\mathcal{N}(A_1)$, one can simply choose a new segment in $\cop$ with respect to its center $\gamma$ being zero. It is easy to verify that the resulting phase interval of $A_1$ obtained from the arc interval of the new segment is contained in $\interval{-\frac{49\pi}{180}}{\frac{49\pi}{180}}$. The new worst case of the sum becomes 
$\interval{-\frac{49\pi}{180}}{\frac{49\pi}{180}} + \Psi(A_2) + \Psi(A_3)\in \interval[open]{-\pi}{\pi}$, and hence the invertibility of $I+A_3 A_2A_1$ can be concluded accordingly.
 \end{example}

 In Example~\ref{example: gamma_phase}, a simple choice of the new segment with the center $\gamma=0$ for one known matrix reduces the conservatism. To subsume such an idea into a general case, we propose an adjustable scalar parameter $\gamma\in\interval[open right]{-\pi}{\pi}$ to represent the unit normal vector $e^{j\gamma}$ of a segment as illustrated by Fig.~\ref{fig:new_angle}. Then, $\gamma$ can be vividly viewed as an \emph{inclined center} of a segment.

\begin{figure}[htb]
 \vspace{-3.5mm}
 \centering
 \includegraphics[width=2.4in, trim={1.8cm 1.7cm 1.5cm 1.2cm}, clip]{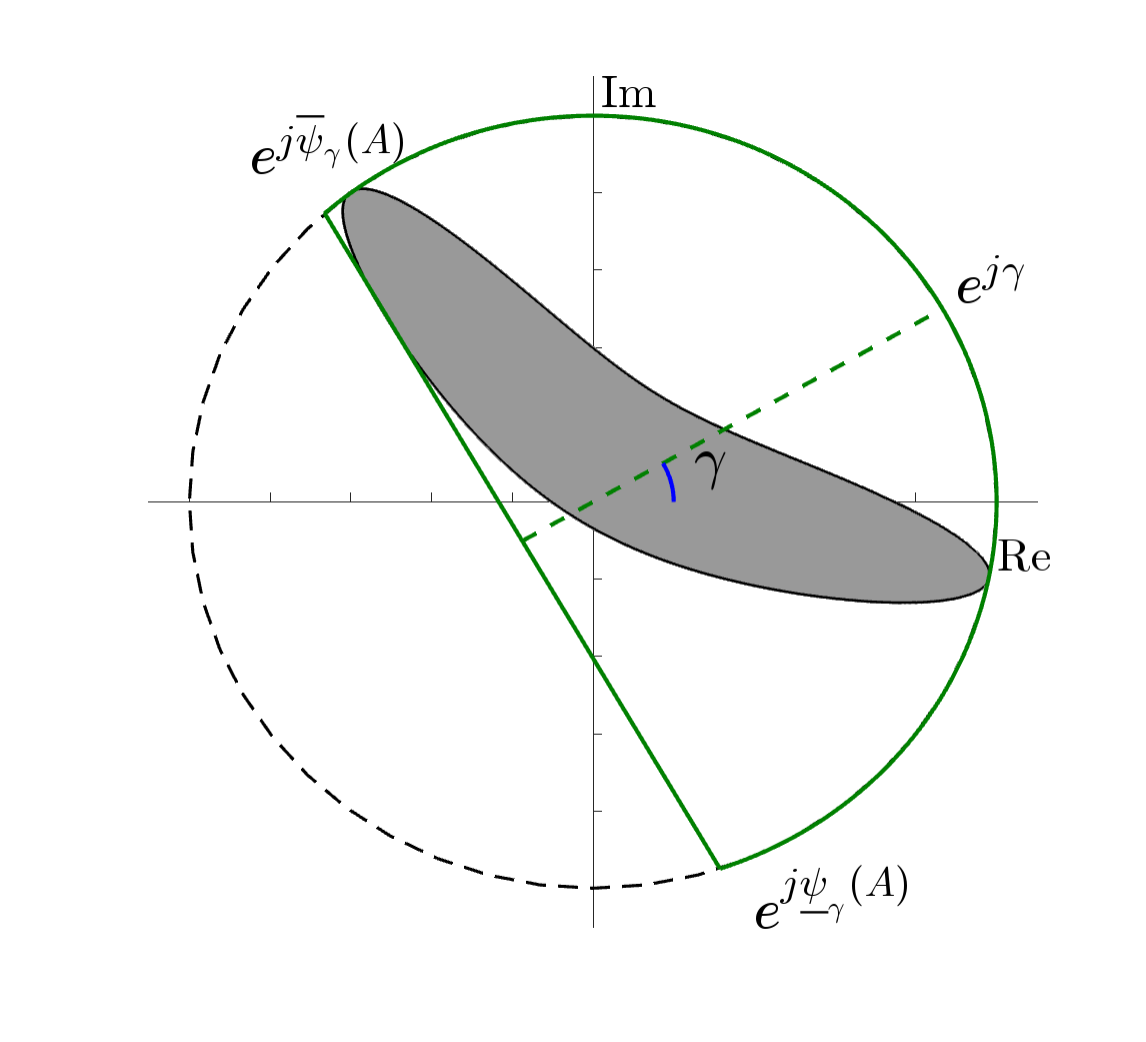}
 \vspace{-1mm}
 \caption{An illustration of $\Psi_\gamma(A)=\interval{\frac{-71.7\pi}{180}}{\frac{131.7\pi}{180}}$ with respect to a given parameter $\gamma=\frac{\pi}{6}$ for the same $A$ in Fig.~\ref{fig:seg_phase}.}\label{fig:new_angle}
\end{figure}
 
For a matrix $A\in \cnn$ and a parameter $\gamma\in\interval[open right]{-\pi}{\pi}$, graphically we tailor \textit{the smallest circular segment} with respect to the \emph{given center $\gamma$} to cover $\mathcal{N}(A)$. As shown in Fig.~\ref{fig:new_angle}, the grey region $\mathcal{N}(A)$ is covered by the green-bordered segment with respect to a given parameter $\gamma=\frac{\pi}{6}$. 
 The \textit{$\gamma$-segmental phase} $\Psi_\gamma(A)$, depending on $\gamma\in \interval[open right]{-\pi}{\pi}$, is defined as the following interval given by the arc edge of the segment:
\be\label{eq:def_gamma_phase}
\Psi_\gamma(A)\coloneqq \interval[scaled]{\underline{\psi}_\gamma(A)}{\overline{\psi}_\gamma(A)},
\ee
where $\underline{\psi}_\gamma(A)\in \interval{\gamma}{\gamma+\pi}$ (resp. $\overline{\psi}_\gamma(A) \in \interval{\gamma-\pi}{\gamma}$) is given by the {argument} of the lower (resp. upper) endpoint of the arc. The algebraic expression of \eqref{eq:def_gamma_phase} can be found in Appendix~\ref{appendix:angle}. Obvious differences exist between the segmental phase $\Psi(A)$ and the $\gamma$-segmental phase $\Psi_\gamma(A)$. The former is an intrinsic notion for $A$ similarly to the gain notion $\overline{\sigma}(A)$. The latter depends on a parameter $\gamma$, and thereby in \eqref{eq:def_gamma_phase} the segment for covering $\mathcal{N}(A)$ always lies along the \emph{unit normal vector} $e^{j\gamma}$. The length of the phase interval $\Psi_\gamma(A)$ is always greater than or equal to that of any phase interval $\Psi \in \Psi(A)$. Particularly, if we set the parameter $\gamma=\gamma^\star \in \gamma^\star(A)$ in \eqref{eq:seg_phase_def}, then $\Psi_\gamma(A) = \Psi$. Finally, note that $\Psi_\gamma(A)$ is {unique} unlike $\Psi(A)$. 

The roles of $\Psi(A)$ and $\Psi_\gamma(A)$ are played differently in Problem~\ref{problem_matrix}, partially explained in Example~\ref{example: gamma_phase}. The former is adopted to characterize phase-bounded uncertain sets, analogously to the matrix gain used for gain-bounded sets. The latter is considered as an angular scaling technique for known matrices in a loop, similarly to gain-scaling ideas. Specifically, given $\alpha, \beta\in \interval[open right]{-2\pi}{2\pi}$, denote by $\mathcal{A}_{\interval{\alpha}{\beta}}^\nn$ the set of \emph{phase-bounded matrices}:
\be\label{eq: phase_uncertain}
 \mathcal{A}_{\interval{\alpha}{\beta}}^\nn\coloneqq \left\{A\in \cnn\mid \Psi(A) \subset \interval{\alpha}{\beta}, \beta-\alpha \in \interval[open right]{0}{2\pi} \right\}.
\ee
There is no ambiguity in adopting the inclusion notation $\Psi(A)\subset \interval{\alpha}{\beta}$ in \eqref{eq: phase_uncertain}: $\Psi(A)$ here must be a singleton since it is impossible to have two selections that are simultaneously contained in $\interval{\alpha}{\beta}$ according to Proposition~\ref{prop: sectorial}.

For $k\in \mathcal{K}$, we assume $A_k \in \mathcal{A}_{\interval{\alpha_k}{\beta_k}}^\nn$; for $k\in \mathcal{K}^\prime$, we can search feasible parameters $\gamma_k\in \interval[open right]{-\pi}{\pi}$ to obtain $\Psi_{\gamma_k}(A_k)$ such that a new small phase condition is satisfied. The idea distilled from Example~\ref{example: gamma_phase} can be rigorously formulated into the following theorem, our solution to Problem~\ref{problem_matrix}:
 
\begin{theorem}\label{thm:matrix_small_phase_cyc}
 For $A_1, A_2, \ldots, A_m \in\cnn$, suppose $A_k\in \mathcal{A}_{\interval{\alpha_k}{\beta_k}}^\nn$ in \eqref{eq: phase_uncertain} for $k\in \mathcal{K}$. The matrix $I+A_mA_{m-1}\cdots A_1$ is invertible if there exist $\gamma_k\in \interval[open right]{-\pi}{\pi}$ for $k\in \mathcal{K}^\prime$ and an integer $l\in \mathbb{Z}$ such that
 \be\label{eq:cyclic_small_phase}
 \sum_{k\in \mathcal{K}^\prime} \Psi_{\gamma_k}(A_k) + \sum_{k\in \mathcal{K}} \interval{\alpha_k}{\beta_k} \subset \interval[open]{-\pi}{\pi} + 2\pi l. 
 \ee
\end{theorem}
\begin{proof}
See Appendix~\ref{appendix:angle}.
\end{proof}

Theorem~\ref{thm:matrix_small_phase_cyc} interlaces the segmental phase with $\gamma$-segmental phase into a single condition \eqref{eq:cyclic_small_phase}. The searching of feasible $\gamma_k$ in \eqref{eq:cyclic_small_phase} can be formulated as a nonlinear optimization problem by exploiting the expression \eqref{eq:def_gamma_phase_2}. Furthermore, one may add another layer of optimization over $\gamma_k$ in the sense that the length of the sum of the phase intervals in \eqref{eq:cyclic_small_phase} is minimized. Such a problem shares a similar flavor with the $\mu$-analysis \cite[Ch.~11]{Zhou:96} by requiring optimization over the set of structured gain-scaling matrices $D$. Roughly speaking, $\overline{\sigma}(A_1)\min_{D}\overline{\sigma}(DA_2D^{-1})<1$ for the invertibility of $I+A_2A_1$. The above feasibility and optimization problems for the segmental phase lie beyond the scope of the current paper and will be explored in our further research. 
 
\subsection{Angular Scaling of Theorem~\ref{thm:segmental_phase_stability}}
We revisit Problem~\ref{problem} by applying the core idea extracted from Theorem~\ref{thm:matrix_small_phase_cyc} to feedback stability analysis. Let $\mathcal{S}_{\interval{\alpha}{\beta}}^{\nn}$ denote the following set of \emph{phase-bounded uncertain stable systems} having no $\jw$-axis zeros:
\begin{align}\label{eq: uncertain_phase_sys_set}
 \hspace{-3mm} \mathcal{S}_{\interval{\alpha}{\beta}}^\nn \coloneqq \left\{P\in \rhinf^{n\times n} \mid { \Psi(P(\jw)) \subset \interval{\alpha(\omega)}{\beta(\omega)}}, \right. \notag \\
 \left. \text{where}~\omega\in \interval{0}{\infty}, \beta(\omega)-\alpha(\omega )\in
 \interval[open right]{0}{2\pi} \right\},
 \end{align}
where $\alpha:\interval{0}{\infty} \to \mathbb{R}$ and $\beta\colon \interval{0}{\infty} \to \mathbb{R}$ jointly characterize a frequency-wise phase bound. Consider a cyclic feedback system consisting of uncertain $P_k\in \mathcal{S}_{\interval{\alpha_k}{\beta_k}}^\nn$ for $k\in \mathcal{K}$ and known $P_k\in \mathcal{P}^\nn$ with the set $\Omega_k$ in \eqref{eq: pole_zero_set} for $k\in \mathcal{K}^\prime$. 

For each $P_k$ with $k\in \mathcal{K}^\prime$, we search for a feasible frequency-wise function $\gamma_k(\omega)\in \mathbb{R}$ for a cyclic loop such that a ``shaped'' small phase condition is satisfied. An extra assumption on $\gamma_k(\omega)$ is required below to guarantee the function being well-behaved. Recall the notation $\Omega= {\bigcup_{k\in \mathcal{K}^\prime} \Omega_k}$ and consider the following set of scalar functions for each $P_k$:
\begin{multline}\label{eq:frequency_gamma_set}
 \hspace{-3mm} \mathcal{F}_k\coloneqq \left\{\gamma_k: \interval{0}{\infty} \to \mathbb{R} \mid \gamma_k(\omega)~\text{is piecewise continuous}; \right. \\
 \left. \gamma_k(\omega)~\text{is continuous when}~\omega\in {\Omega\setminus\Omega_k}\right\}. 
 \end{multline}
For each $P_k\in \mathcal{P}^\nn$, the set $\mathcal{F}_k$ is rather general except on a mild condition of continuity at $\omega \in\Omega\setminus\Omega_k$. The purpose of the continuity is to assure that when $\omega \in\Omega\setminus\Omega_k$, i.e., the frequencies in which other subsystems have poles/zeros, no extra phase-value jump can be caused from $\mathcal{F}_k$ for all $P_k(s)$.

Here comes our second feedback stability result to Problem~\ref{problem}, an angular-scaling version of Theorem~\ref{thm:segmental_phase_stability}. 

\begin{theorem}\label{thm:small_segmental_phase}
For $P_1, P_2, \ldots, P_m \in \mathcal{P}^\nn$, assume $P_k\in \mathcal{S}_{\interval{\alpha_k}{\beta_k}}^\nn$ in \eqref{eq: uncertain_phase_sys_set} for $k\in \mathcal{K}$. The cyclic feedback system is stable if there exist functions $\gamma_k\in \mathcal{F}_k$ in \eqref{eq:frequency_gamma_set} for $k\in \mathcal{K}^\prime$ such that for all $\omega\in \interval{0}{\infty}$, 
 \be\label{eq:seg_small_phase}
 \sum_{k\in \mathcal{K}^\prime} \Psi_{\gamma_k(\omega)}(P_k(\jw)) + \sum_{k\in \mathcal{K}} \interval{\alpha_k(\omega)}{\beta_k(\omega)} \subset \interval[open]{-\pi}{\pi}.
 \ee
 \end{theorem}
\begin{proof}
 See Appendix~\ref{sec:proofs}.
\end{proof}

There is a notable gain/phase correspondence between Theorem~\ref{thm:small_segmental_phase} and Lemma~\ref{lem:small_gain_cyc} for uncertain cyclic feedback systems. Showing the existence of $\gamma_k(\omega)$ boils down to a frequency-sweeping test. How to efficiently solve appropriate $\gamma_k(\omega)$ inspired by the $\mu$-analysis remains nontrivial and is our ongoing research, as briefly discussed in the matrix case earlier.

\section{Conclusions and Future Works}\label{sec:conclusion}
 In this paper, we first proposed the segmental phase of matrices and MIMO LTI systems. The segmental phase, acting as a new counterpart to the gain notion, has the crucial product eigen-phase bound. Based on the segmental phase, we then established a small phase theorem for stability analysis of cyclic feedback systems with multiple subsystems, which stands side-by-side with the small gain theorem. We further established a mixed small gain/phase theorem by interlacing frequency-wise gain/phase conditions. Lastly, when some of subsystems in the loop are known, we proposed the $\gamma$-segmental phase and developed angular scaling techniques to reduce the conservatism of the main result.

 It is our hope that this paper offers a brand-new viewpoint for the recent renaissance of phase in the field of systems and control. Extensions of the segmental phase to nonlinear time-varying systems and efficient computation of segmental phases based on nonlinear programming are under investigation. Other future research directions embrace a counterpart of $\hinf$-controller synthesis methods \cite{Zhou:96} based on the segmental phase. In addition, the cyclic feedback system considered throughout is restricted to be a single-loop system. A general large-scale network may contain multiple loops, for which multi-loop cyclic small gain theorems \cite{Jiang:18, Liu:11, Marelli:23, ChenKW:24} have been successfully established. Extending our single-loop analysis to a multi-loop scenario is our ongoing research.
 
\appendices

\section{The Segmental Phase Behind the Scenes}\label{appendix:phase}
The purpose of this appendix is to facilitate the understanding of the segmental phase by establishing an important connection with the matrix singular angle and some optimization formulations. After having these preparations, we provide the full proof of Theorem~\ref{lem:eigen_phase}. 
\subsubsection{The Matrix Case}\label{appendix:phase_A} 
We first present some preliminaries on the notion of singular angle \cite[Sec.~23]{Wielandt:67}, an old but less known concept. For a matrix $A\in \mathbb{C}^{n\times n}$, the \textit{singular angle} $\theta(A)\in \interval{0}{\pi}$ is defined by
\be\label{eq:matrix_angle}
\theta(A)\coloneqq \displaystyle \sup_{\substack{0\neq x\in \mathbb{C}^n, {Ax}\neq 0}} \arccos \frac{\rep \rbkt{x^* Ax}}{\abs{x}\abs{Ax}}, 
\ee
where $\rep\rbkt{\cdot}$ represents the real part of a scalar. 
 By exploiting the normalized numerical range $\mathcal{N}(A)$ in \eqref{eq:matrix_NNR}, we reformulate definition (\ref{eq:matrix_angle}) as follows:
$\theta(A)= \sup_{z \in \mathcal{N}(A)} \arccos \rep \rbkt{z}$.
 The singular angle has two useful properties as detailed in the following lemmas {\cite[Sec.~23.5]{Wielandt:67}}:
\begin{lemma}\label{lem: angle_bound}
 For a matrix $A\in \cnn$, it holds that
$\abs{\angle\lambda_i(A)}\leq \theta(A)$
for $i=1, 2, \ldots, n$, where $\lambda_i(A) \neq 0$.
\end{lemma}
\begin{lemma}\label{lem:matrix_triangle_ineq}
 For matrices $A$, $B\in \mathbb{C}^{n\times n}$, it holds that $\theta(AB)\leq \theta(A)+\theta(B)$.
\end{lemma}

Equipped with the singular angle, we establish an algebraic expression of the segmental phase defined in \eqref{eq:seg_phase_def}. For $A\in \cnn$, a smallest segment $\mathcal{S}$ for covering $\mathcal{N}(A)$ can be uniquely decided from its center and arc edge. Due to \eqref{eq: rotation_NNR}, the arc edge of $\mathcal{S}$ and that of $e^{-j\gamma}\mathcal{S}$ for covering $\mathcal{N}(e^{-j\gamma}A)$ share the same length for $\gamma\in \interval[open right]{-\pi}{\pi}$. The length can be computed via the singular angle $\theta(e^{-j\gamma}A)$. Specifically, the {phase center} $\gamma^\star(A)\in \interval[open right]{-\pi}{\pi}$ of $A$ can be expressed through the following optimization problem: 
\be\label{eq:matrix_phase_center}
\gamma^\star(A) = \arg \min_{\gamma\in \interval[open right]{-\pi}{\pi}} \theta \rbkt{e^{-j\gamma}A}
\ee
which may be multi-valued. Then, $r^{\star}(A)\coloneqq \theta \rbkt{e^{-j\gamma^\star(A)}A} \in \interval{0}{\pi}$ is termed the \textit{phase radius} of $A$.
The segmental phase of $A$ can be represented by
\bex
 \Psi(A)= \interval{\underline{\psi}(A)}{\overline{\psi}(A)},\eex
where
\begin{equation}\label{eq:seg_phase_sector}
 \begin{aligned}
 \overline{\psi}(A)&= \gamma^\star(A)+ r^{\star}(A), \\ \underline{\psi}(A)&= \gamma^\star(A)-r^{\star}(A).
 \end{aligned}
\end{equation}
The phase radius $r^\star(A)$ is used to characterize the length of the shortest arc edge in Fig.~\ref{fig:seg_phase}, i.e., $
\overline{\psi}(A)-\underline{\psi}(A)=2r^\star(A)$, which has an analogy to the condition number $\kappa(A)={\overline{\sigma}(A)}/{\underline{\sigma}(A)}$. In particular, when $A$ is positive definite, $r^\star(A)$ and $\kappa(A)$ are closely connected as stated in Example~\ref{example2}\ref{item:prop_3}. Based on the above connections, we are ready to prove Theorem~\ref{lem:eigen_phase}. 
 
\emph{Proof of Theorem~\ref{lem:eigen_phase}}:~Let $A\coloneqq A_mA_{m-1}\cdots A_1$. For each phase interval selection $\Psi_k\in \Psi(A_k)$ with $k=1,2,\ldots, m$, denote the corresponding phase center selection by $\gamma_k^{\star} \in \gamma^\star(A_k)$. We only need to consider the case that $\sum_{k=1}^{m} \Psi_k $ is contained in an open $2\pi$-interval $\interval[open]{\alpha-\pi}{\alpha+\pi}$, where $\alpha\in \mathbb{R}$, since otherwise the statement trivially holds for $\angle\lambda_i(A)$. In addition, $\angle\lambda_i(A)$ can be chosen modulo $2\pi$. By hypothesis and using \eqref{eq:seg_phase_sector}, we have that
 \begin{align*}
 \textstyle \sum_{k=1}^{m} \gamma_k^{\star} + \theta\rbkt{e^{-j\gamma_k^{\star}} A_k} &< \alpha+\pi,\\
 \textstyle \sum_{k=1}^{m} \gamma_k^{\star} - \theta\rbkt{e^{-j\gamma_k^{\star}} A_k} &>\alpha-\pi.
 \end{align*}
Since $\theta\rbkt{e^{-j\gamma_k^{\star}} A_k}\geq 0$ for all $k=1,2,\ldots, m$, it follows that
 $
 \sum_{k=1}^{m} \gamma_k^{\star} \in \interval[open]{\alpha-\pi}{\alpha+\pi}$.
 Note that for arbitrary $\gamma\in \interval[open]{\alpha-\pi}{\alpha+\pi}$, it holds that 
 \be\label{eq:thm_eigen_phase2}
 \angle\lambda_i(A)-\gamma=\angle \lambda_i\rbkt{e^{-j\gamma}A}\mod 2\pi
 \ee
 for $\lambda_i(A)\neq 0$ and $i=1, 2, \ldots, n$. In addition, according to Lemma~\ref{lem: angle_bound}, we have that \be\label{eq:thm_eigen_phase3}
 \abs{\angle \lambda_i\rbkt{e^{-j\gamma}A}} \leq \theta\rbkt{e^{-j\gamma}A}\ee
 for $\lambda_i(A)\neq 0$ and $i=1, 2, \ldots, n$. 
 Combining \eqref{eq:thm_eigen_phase2} and \eqref{eq:thm_eigen_phase3} and substituting $\gamma=\sum_{k=1}^{m} \gamma_k^{\star}$ into them yield that
 \begin{align*}
 &\textstyle \abs{\angle\lambda_i(A)- \sum_{k=1}^{m} \gamma_k^{\star}} \leq \theta\rbkt{e^{-j{\sum_{k=1}^{m} \gamma_k^{\star}}}A}\\
 =&~\theta\rbkt{e^{-j{\gamma_m^{\star}}}A_m e^{-j{\gamma_{m-1}^{\star}}}A_{m-1} \cdots e^{-j{\gamma_1^{\star}}}A_1 } \\
 \leq &~ \textstyle \sum_{k=1}^{m} \theta\rbkt{e^{-j\gamma^\star_k}A_k}
 \end{align*}
 for $\lambda_i(A)\neq 0$ and $i=1, 2, \ldots, n$,
 where the last inequality uses Lemma~\ref{lem:matrix_triangle_ineq}.
 This implies that 
 $\sum_{k=1}^{m} \gamma_k^{\star} - \theta\rbkt{e^{-j\gamma_k^{\star}} A_k}\leq \angle\lambda_i(A)\leq \sum_{k=1}^{m} \gamma_k^{\star} + \theta\rbkt{e^{-j\gamma_k^{\star}} A_k}$
 which is equivalent to that 
 $\angle \lambda_i(A) \in \sum_{k=1}^{m} \Psi_k$ 
 for $\lambda_i(A)\neq 0$ and $i=1, 2, \ldots, n$. \hspace*{\fill}~\QED

\subsubsection{The System Case}\label{appendix: phase_sys} 
The major difference between the matrix and system segmental phases lies in the choice of suitable principal branches for their phase centers. For the system case, we have shown in Section~\ref{sec:system_segphase_def} that $\Gamma:\interval{0}{\infty}\to \mathcal{I}^{2\pi}$ can be adopted for taking possible values of phase centers, with $\mathcal{I}^{2\pi}=\{\interval[open right]{\alpha}{\alpha+2\pi}\mid \alpha\in \mathbb
 R\}$ being a set of principal branches. For $P\in \mathcal{P}^\nn$ in \eqref{eq: system_class}, the \textit{phase center} $\gamma^{\star}(P(\jw))$ and the \textit{phase radius} $r^{\star}(P(\jw))$ can be represented by
 \begin{align}
 \gamma^{\star}(P(\jw)) &= \arg \min_{\gamma\in \Gamma(\omega)} \theta \rbkt{e^{-j\gamma}P(\jw)},\label{eq:phase_center_response} \\
 r^{\star}(P(\jw)) &\coloneqq \theta \rbkt{e^{-j\gamma^\star(P(\jw))}P(\jw)}\in \interval{0}{\pi}\notag 
 \end{align}
 for $\omega\in \interval{0}{\infty}\setminus \Omega$, respectively. When solving \eqref{eq:phase_center_response} frequency-wise, we gradually shift $\alpha$ to designate a new branch $\interval[open right]{\alpha}{\alpha+2\pi}$ such that within the branch $\gamma^\star(P(\jw))\in \interval[open right]{\alpha}{\alpha+2\pi}$ is unique and thus continuous for $\omega\in \interval{0}{\infty}\setminus\Omega$. Then, the segmental phase of $P(s)$ in \eqref{eq:system_phase} has the following expression:
 \bex \Psi(P(\jw))=\interval[scaled]{ \underline{\psi}(P(\jw))}{\overline{\psi}(P(\jw))},
 \eex
 where
 \begin{equation}\label{eq:phase_response}
 \begin{aligned}
 \overline{\psi}(P(\jw))= \gamma^\star(P(\jw))+ r^{\star}(P(\jw)),\\
 \underline{\psi}(P(\jw))= \gamma^\star(P(\jw))-r^{\star}(P(\jw)).
 \end{aligned}
 \end{equation}
 
\section{The $\gamma$-Segmental Phase Behind the Scenes}\label{appendix:angle}
For a matrix $A\in \cnn$ and a parameter $\gamma\in\interval[open right]{-\pi}{\pi}$, the \emph{$\gamma$-segmental phase} $\Psi_\gamma(A)$ in \eqref{eq:def_gamma_phase} can be represented by $\Psi_\gamma(A)= \interval{\underline{\psi}_\gamma(A)}{\overline{\psi}_\gamma(A)}$,
where
\begin{equation}\label{eq:def_gamma_phase_2}
 \begin{aligned}
 \overline{\psi}_\gamma(A)&\coloneqq \gamma+\theta \rbkt{e^{-j\gamma}A},\\
 \underline{\psi}_\gamma(A)&\coloneqq \gamma-\theta \rbkt{e^{-j\gamma}A}.
 \end{aligned}
\end{equation}

The full proof of Theorem~\ref{thm:matrix_small_phase_cyc} is provided below.

\emph{Proof of Theorem~\ref{thm:matrix_small_phase_cyc}}:~For the notational brevity, denote $\gamma_k^{\star}\coloneqq \gamma^\star(A_k)$ for $k\in \mathcal{K}$. Without loss of generality, consider $l=1$.
 Note that condition \eqref{eq:cyclic_small_phase} implies that $\textstyle \sum_{k\in \mathcal{K}^\prime} \Psi_{\gamma_k}(A_k) + \sum_{k\in \mathcal{K}} \Psi(A_k) \subset \interval[open]{-\pi}{\pi}$
 which is equivalent to that
 \begin{align*}
 \textstyle \sum_{k\in \mathcal{K}^\prime } \sbkt{\gamma_k + \theta(e^{-j\gamma_k} A_k)} + \sum_{k\in \mathcal{K}} \sbkt{\gamma^\star_k+ r^{\star}(A_k)} &<\pi,\\
 \textstyle \sum_{k\in \mathcal{K}^\prime } \sbkt{\gamma_k - \theta(e^{-j\gamma_k} A_k)} + \sum_{k\in \mathcal{K}} \sbkt{\gamma^\star_k- r^{\star}(A_k)} &>-\pi.
 \end{align*}
 The above two inequalities can be compactly rewritten as:
 \begin{multline}\label{eq:rotated_angle_condition}
 \textstyle \sum_{k\in \mathcal{K}^\prime} \theta(e^{-j\gamma_k} A_k)
 + \sum_{k\in \mathcal{K}} \theta(e^{-j\gamma^{\star}_k} A_k) + \\ \textstyle \abs{ \sum_{k\in \mathcal{K}^\prime} \gamma_k +\sum_{k\in \mathcal{K}} \gamma^\star_k } <\pi.
 \end{multline}
 In addition, rewrite $A_k$ as $A_k= \rbkt{e^{-j\gamma_k^{\star}}A_k}e^{j\gamma_k^{\star}}$ for $k\in \mathcal{K}$ and $A_k= \rbkt{e^{-j\gamma_k}A_k}e^{j\gamma_k} $ for $k\in \mathcal{K}^\prime$. Then for the product form ${A_mA_{m-1}\cdots A_1}$, putting the scalars $e^{j\gamma_k^{\star}}$ and $e^{j\gamma_k}$ together, i.e., $e^{j \rbkt{\sum_{k\in \mathcal{K}^\prime} \gamma_k+ \sum_{k\in \mathcal{K}}\gamma_k^{\star}}}$, and applying Lemma~\ref{lem:matrix_triangle_ineq} to $A\coloneqq A_mA_{m-1}\cdots A_1$ yield that
 \begin{align*}
 \theta(A) \leq & \textstyle \sum_{k\in \mathcal{K}^\prime} \theta(e^{-j\gamma_k} A_k)
 + \sum_{k\in \mathcal{K}} \theta(e^{-j\gamma^{\star}_k} A_k)\\
 & + \theta\left(e^{j\rbkt{\sum_{k\in \mathcal{K}^\prime}\gamma_k +\sum_{k\in \mathcal{K}}\gamma_k^{\star} }}\right)\\
 = & \textstyle \sum_{k\in \mathcal{K}^\prime} \theta(e^{-j\gamma_k} A_k)
 + \sum_{k\in \mathcal{K}} \theta(e^{-j\gamma^{\star}_k} A_k) \\
 &+ \textstyle \abs{ \sum_{k\in \mathcal{K}^\prime} \gamma_k +\sum_{k\in \mathcal{K}} \gamma^\star_k }<\pi,
 \end{align*}
 where the last inequality follows from \eqref{eq:rotated_angle_condition}. 
 By Lemma~\ref{lem: angle_bound}, it holds that $\abs{\angle \lambda_i(A)}\leq \theta(A)<\pi$ for $\lambda_i(A)\neq 0$ and $i=1, 2, \ldots, n$.
 This gives that $I+A$ is invertible. \hspace*{\fill}~\QED
 
 \section{}\label{sec:proofs}
This appendix provides the proofs of Theorems~\ref{thm:segmental_phase_stability} and \ref{thm:small_segmental_phase} in order. Before we dive into the details, it is worthy of foreshadowing the underlying idea of the proofs based on the generalized Nyquist criterion \cite{Desoer:80}. Let us temporarily look at a simpler case when all $P_k\in \rhinf^\nn$ with $\Omega_k=\emptyset$. The road map is organized as follows. The small phase condition~\eqref{eq:thm:seg_phase_stability} will imply that $\lambda_i(P_m(\jw)P_{m-1}(\jw)\cdots P_1(\jw))\subset \interval[open]{-\pi}{\pi}$ for all $\omega \in \interval{-\infty}{\infty}$. This indicates that the eigenloci of $P_m(s)P_{m-1}(s)\cdots P_1(s)$ are all contained in a simply connected region, the $\interval[open]{-\pi}{\pi}$-cone. The number of encirclements of ``$-1$'' made by the closed paths formed by the eigenloci is thus zero. The closed-loop stability then follows from \cite{Desoer:80}.

We are ready to show the proofs for semi-stable systems where the indented contour $\mathbf{NC}$ in Fig.~\ref{fig:nyquist} will be needed for addressing $\jw$-axis poles/zeros, while the basic idea keeps the same as above. We first show the proof of Theorem~\ref{thm:segmental_phase_stability}.

\emph{Proof of Theorem~\ref{thm:segmental_phase_stability}}:~Denote by $P\coloneqq P_m P_{m-1}\cdots P_1$ and recall that $\Omega^p = \bigcup_{k=1}^{m} \Omega_k^p$ and $\Omega^z = \bigcup_{k=1}^{m} \Omega_k^z$. Since there is no unstable pole-zero cancellation, according to Lemma~\ref{lem:feedback_no_cancellation}, it suffices to show that $\rbkt{I+P}^{-1}\in \rhinf^{n\times n}$. Moreover, under Assumption~\ref{assum: system}, the $\jw$-axis poles $j\Omega^p$ and zeros $j\Omega^z$ of $P(s)$ are exactly dual cases. We only present a detailed proof involving the poles since the zeros can be addressed similarly by replacing the phase decreasing-shift by increasing-shift.
 The proof will be divided into \emph{two steps}. Firstly, we show that $\det \sbkt{I+P(s)}\neq 0$ for all $s$ encircled by the indented contour $\mathbf{NC}$ with a semicircular detour around every $s=\pm j\omega_0$, where $\omega_0 \in \Omega^p$. Secondly, we show that any open-loop pole at $s= j\omega_0$ is not a pole of the cyclic loop, where $\omega_0 \in \Omega^p$.

 \emph{Step~1}: Applying Theorem~\ref{lem:eigen_phase} with the open interval $\interval[open]{-\pi}{\pi}$ to condition~\eqref{eq:thm:seg_phase_stability} in a frequency-wise manner\footnote{In Theorem~\ref{lem:eigen_phase}, the principal branch of the matrix phase interval $\Psi(A_k)$ is fixed, while the system phase $\Psi(P_k(\jw))$ as a function of $\omega$ can be any interval bounded by $2\pi$ on the whole real line. To apply Theorem~\ref{lem:eigen_phase} to \eqref{eq:thm:seg_phase_stability} frequency-wise, we need a slight extension of Theorem~\ref{lem:eigen_phase} with possibly different principal branches of the phase interval for $A_k$. This can be easily done by extending the matrix phase interval \eqref{eq:seg_phase_def} to a general interval, namely, $\Psi^{\mathrm{g}}(A)\coloneqq \Psi(A)+ 2\pi l$, where $l\in \mathbb{Z}$ is any integer.} gives that
 \begin{equation}\label{eq:thm_eigen_seg_phase}
 \textstyle\angle \lambda_i (P(s)) \in \sum_{k=1}^{m} \Psi(P_k(s)) \subset \interval[open]{-\pi}{\pi}
 \end{equation}
 for all nonzero $\lambda_i(P(s))$, $s=j\omega$ and $i=1, 2, \ldots, n$, where $\omega\in \interval{0}{\infty}\setminus \Omega^p$. We need to show the case of $\angle \lambda_i(P(s))\in \interval[open]{-\pi}{\pi}$ when $s$ moves along the semicircular indentations. 
 
 For any $j\omega_0\in j{\Omega^p}$ of order $l$, let $s=j\omega_0 +\epsilon e^{j\alpha}$, where $\epsilon>0$ is sufficiently small and $\alpha \in \interval{-\frac{\pi}{2}}{\frac{\pi}{2}}$ that is, $s\in \semicir$. Due to Assumptions~\ref{assum: system} and \ref{assum: singleton}, notice that $\sum_{k=1}^{m} \Psi(P_k(s))$ has possible phase-value jumps along $s\in \semicir$ which can only come from two cases:
 \begin{enumerate}
 \item [(a)] \label{item:case11} The pole at $s=\jw_0$ of order $l$;
 \item [(b)] For those $P_k(s)$ having no pole at $s=j\omega_0$, $P_k(s)$ has a zero at $s=\jw_0$. 
 \end{enumerate}
 It is clear that $\jw_0$ cannot be a zero of any $P_k(s)$ since unstable pole-zero cancellation does not exist, and hence a phase-value jump can only stem from case (a). Furthermore, by the small phase condition \eqref{eq:thm:seg_phase_stability}, the pole $j\omega_0$ is at most of order $2$, since otherwise it can generate a phase-value jump greater than $2\pi$ along $s\in \semicir$ which breaks \eqref{eq:thm:seg_phase_stability}. Without loss of generality, we have the following three possibilities for the pole $\jw_0$:
 \begin{enumerate}
 \renewcommand{\theenumi}{\textup{(\roman{enumi})}}\renewcommand{\labelenumi}{\theenumi}
 \item \label{item:case1}
 It is from $P_1(s)$ and the order $l=1$;
 \item \label{item:case2}
 It is from $P_1(s)$ and the order $l=2$;
 \item \label{item:case3} It consists of two single poles from $P_1(s)$ and $P_2(s)$ and the total order $l=2$.
 \end{enumerate}

 For $s\in \mathbf{SC}(\epsilon, \jw_0)$ and $\alpha \in \interval{-\frac{\pi}{2}}{\frac{\pi}{2}}$, note the following partial fraction expansion of $P_1(s)$ at $j\omega_0$:
 \begin{equation}\label{eq:thm_pole_seg_expansion_P}
 \begin{aligned}
 \hspace{-2mm} P_1(s)&= \textstyle\frac{H_1}{(s-j\omega_0)^2} +\frac{K_1}{s-\jw_0}+R_1(s) \\
 &= H_1\epsilon^{-2}e^{-j2\alpha}+ K_1 \epsilon^{-1}e^{-j\alpha} + R_1(\jw_0+\epsilon e^{j\alpha}), 
 \end{aligned}
 \end{equation}
where the coefficients $H_1$ and $K_1$ are constant matrices, $R_1(s)$ is analytic at $s=j\omega_0$. In what follows, for Cases~\ref{item:case1}-\ref{item:case3}, as an intermediate but key step we respectively show that $\sum_{k=1}^{m} \Psi(P_k(s)) \subset \interval[open]{-\pi}{\pi}$ also holds for all $s\in \semicir$, which complements \eqref{eq:thm:seg_phase_stability} along $\mathbf{NC}$. For brevity, for $s\in \semicir$, denote two interval-valued functions:
\begin{align*}
 \textstyle f(\alpha, \epsilon)\coloneqq \sum_{k=1}^{m}\Psi(P_k(s))~\text{and}~\widehat{\Psi}(s)\coloneqq \sum_{k=2}^{m}\Psi(P_k(s)).
\end{align*}

 Case \ref{item:case1}: $H_1=0$ and by Assumption~\ref{assum: system}, $K_1$ has full rank in \eqref{eq:thm_pole_seg_expansion_P}. For $s\in \semicir$, using \eqref{eq:thm_pole_seg_expansion_P}, we have
 \begin{equation}\label{eq:thm_f_alpha_and_epsilon_seg}
 \begin{aligned}
 &~f(\alpha, \epsilon)=\Psi(P_1(s))+\widehat{\Psi}(s)\\
 =&~\Psi(K_1\epsilon^{-1}e^{-j\alpha}+R_1(\jw_0+\epsilon e^{j\alpha}))+ \widehat{\Psi}(\jw_0+\epsilon e^{j\alpha})\\
 =&~\Psi(K_1 +\epsilon R_1(\jw_0+\epsilon e^{j\alpha})e^{j\alpha})-\interval{\alpha}{\alpha}+\widehat{\Psi}(\jw_0+\epsilon e^{j\alpha}),
\end{aligned}
\end{equation}
where the third equality is due to $\epsilon>0$. By condition~\eqref{eq:thm:seg_phase_stability}, when $\alpha =\pm \frac{\pi}{2}$, it holds that $f(\alpha, \epsilon) \subset \interval[open]{-\pi}{\pi}$, i.e., 
\begin{align}
&\hspace{-1mm}\Psi(K_1+\epsilon jR_1(\jw_0+ j\epsilon))+\widehat{\Psi}(\jw_0+j\epsilon)\subset \textstyle \interval[scaled, open]{-\frac{\pi}{2}}{\frac{3\pi}{2}}, \label{eq:thm_seg_residue_case1_eq1}\\
& \hspace{-1mm} \Psi(K_1-\epsilon jR_1(\jw_0-j\epsilon))+\widehat{\Psi}(\jw_0-j\epsilon)\subset \textstyle \interval[scaled, open]{-\frac{3\pi}{2}}{\frac{\pi}{2}} \label{eq:thm_seg_residue_case1_eq2}
\end{align}
for sufficiently-small $\epsilon >0$ according to \eqref{eq:thm_f_alpha_and_epsilon_seg}. The phase values $\Psi(P_k(s))$ are continuously assigned on $s\in \semicir$ for $k=1, 2, \ldots, m$, since $K_1$ has full rank and the phase value $\widehat{\Psi}(s)$ is continuous and unique for all $s\in \interval{j\omega_0-j\epsilon}{j\omega_0+j\epsilon}$ by Assumption~\ref{assum: singleton}. In other words, there is no other phase-value jump on $s\in \semicir$. It then follows that
\begin{align*}
 \lim_{\epsilon\to 0} \textrm{LHS}~\text{of}~\eqref{eq:thm_seg_residue_case1_eq1} \subset { \textstyle \interval[scaled]{-\frac{\pi}{2}}{\frac{3\pi}{2}}},\;
\lim_{\epsilon\to 0} \textrm{LHS}~\text{of}~\eqref{eq:thm_seg_residue_case1_eq2} \subset \textstyle \interval[scaled]{-\frac{3\pi}{2}}{\frac{\pi}{2}}.
\end{align*}
This implies that $\Psi(K_1)+\widehat{\Psi}(\jw_0)$ must be contained in the intersection of the above two intervals, namely,
\be\label{eq:thm_seg_residue_case1}
\Psi(K_1)+ \widehat{\Psi}(\jw_0) \subset \textstyle \interval[scaled]{-\frac{\pi}{2}}{\frac{\pi}{2}}, 
\ee
 as the phase value $\Psi(K_1)+\widehat{\Psi}(\jw_0)$ is also determined by the continuity of $\Psi(P_k(s))$. In addition, when $\alpha\in \interval[open]{-\frac{\pi}{2}}{\frac{\pi}{2}}$, it follows from \eqref{eq:thm_f_alpha_and_epsilon_seg} and \eqref{eq:thm_seg_residue_case1} that
\begin{align*}
\lim_{\epsilon \to 0} f(\alpha, \epsilon) &= \Psi(K_1)-\interval{\alpha}{\alpha}+\widehat{\Psi}(\jw_0)\\
&=\left[ \textstyle\underline{\psi}(K_1)+\sum_{k=2}^{m}\underline{\psi}\rbkt{P_k(\jw_0)}-\abs{\alpha}, \right.\\
& \hspace{4mm} \left. { \textstyle\overline{\psi}(K_1)+ \sum_{k=2}^{m}\overline{\psi}\rbkt{P_k(\jw_0)} +\abs{\alpha}} \right]\subset \interval[open]{-\pi}{\pi}.
\end{align*}
Combining the above set relation and condition \eqref{eq:thm:seg_phase_stability} yields that $f(\alpha, \epsilon)=\textstyle \sum_{k=1}^{m}\Psi(P_k(s)) \subset \interval[open]{-\pi}{\pi}$ for all $s\in \semicir$.

Case~\ref{item:case2}: We follow similar lines of reasoning as in Case~\ref{item:case1} and omit all repeated arguments. For $s\in \semicir$, we have 
 \begin{equation}\label{eq:thm_falpha_epsilon_seg_case2}
 \begin{aligned}
 f(\alpha, \epsilon)
 =&~\Psi(H_1+ \epsilon K_1 e^{j\alpha} + \epsilon^2 R_1(\jw_0+\epsilon e^{j\alpha})e^{j2\alpha}) \\
 &-\interval{2\alpha}{2\alpha}+\widehat{\Psi}(\jw_0+\epsilon e^{j\alpha}).
\end{aligned}
\end{equation}
Following the similar argument and full rankness of $H_1$ gives 
\begin{align}\label{eq:thm_seg_residue_case2}
 \Psi(H_1)&+\widehat{\Psi}(\jw_0)= 0. 
\end{align}
When $\alpha\in \interval[open]{-\frac{\pi}{2}}{\frac{\pi}{2}}$, it follows from \eqref{eq:thm_falpha_epsilon_seg_case2} and \eqref{eq:thm_seg_residue_case2} that
\begin{align*}
\lim_{\epsilon \to 0} f(\alpha, \epsilon) &=\Psi(H_1)-\interval{2\alpha}{2\alpha} +\widehat{\Psi}(\jw_0)\subset \interval[open]{-\pi}{\pi}.
\end{align*}

Case~\ref{item:case3}: It can be shown by using analogous arguments as in Case~\ref{item:case2}, except for the following differences. First, in addition to \eqref{eq:thm_pole_seg_expansion_P} with $H_1=0$, we need the expansion of 
 $P_2(s)=\frac{K_2}{s-\jw_0}+R_2(s)$, 
where $K_2$ has full rank and $R_2(s)$ is analytic at $s=j\omega_0$. Second, instead of \eqref{eq:thm_seg_residue_case2}, we arrive at the following condition: 
\be\label{eq:thm_seg_residue_case3}
\textstyle \Psi(K_1)+\Psi(K_2)+ \sum_{k=3}^{m} \Psi(P_k(\jw_0))= 0.
\ee

Therefore, for all Cases~\ref{item:case1}-\ref{item:case3}, applying the same argument as in \eqref{eq:thm_eigen_seg_phase} to the results that $\textstyle \sum_{k=1}^{m}\Psi(P_k(s)) \subset \interval[open]{-\pi}{\pi}$ for all $s\in \semicir$, we conclude that
\be\label{eq:thm_eigen_semicir_seg_phase}
 \angle \lambda_i (P(s))
\subset \interval[open]{-\pi}{\pi}
\ee
for all nonzero $\lambda_i(P(s))$, $s\in \semicir$ and $i=1, 2, \ldots, n$. Combining \eqref{eq:thm_eigen_seg_phase} and \eqref{eq:thm_eigen_semicir_seg_phase} and using the conjugate symmetry of $P(s)$ yield that $\abs{\angle\lambda_i(P(s))}<\pi$ for all $s\in \mathbf{NC}$. This further implies that $\det \sbkt{I+P(s)} \neq 0$ for all $s\in \mathbf{NC}$. Additionally, since $\abs{\angle\lambda_i(P(s))}<\pi$ for all $s\in \mathbf{NC}$, the eigenloci of $P(s)$ on $s\in \mathbf{NC}$ never intersect with the negative real axis; i.e., there is no $s\in \mathbf{NC}$ and no $i$ such that $\lambda_i(P(s))=\pi$. This means that the number of encirclements of ``$-1$'' made by the closed paths formed by the eigenloci of $P(s)$ along the contour $\mathbf{NC}$ is zero. By the generalized Nyquist criterion\cite{Desoer:80}, it holds that $\det \sbkt{I+P(s)} \neq 0$ for all $s\in \ccp \cup \{\infty\} \setminus \{ j{\Omega_{\pm}^p}\}$.
 
\emph{Step~2}: It remains to show that any open-loop pole of $P(s)$ at $s=\jw_0$ with $\omega_0\in \Omega^p$ is not a pole of $\rbkt{I+P(s)}^{-1}$. Since we have two cases:
$P(s)= \frac{K}{s-j\omega_0} + R(s)$,
with $K$ full rank or $P(s) = \frac{H}{(s-j\omega_0)^2} + \frac{K}{s-j\omega_0} + R(s)$,
with $H$ full rank. Here in both cases, $R(s)$ has no pole at $j\omega_0$. In the first case, we have $(s-j\omega_0) (I+P(s)) = K + (s-j\omega_0)(I+R(s))$
which is full rank when evaluated at $j\omega_0$. In the second case, $(s-j\omega_0)^2 (I+P(s)) = H + (s-j\omega_0)K+(s-j\omega_0)^2(I+R(s))$
which is full rank when evaluated at $j\omega_0$. Consequently, in both cases, $I+P(s)$ has no zero at $j\omega_0$. Thus $\rbkt{I+P(s)}^{-1}$ has no pole at $j\omega_0$. This completes the proof. \hspace*{\fill}~\QED

The following proof is largely based on the similar arguments as those stated in the proof of Theorem~\ref{thm:segmental_phase_stability}. Thus we only note their significant differences for brevity.

\emph{Proof of Theorem~\ref{thm:small_segmental_phase}}:~Let $P= P_m P_{m-1}\cdots P_1$. 
 Note that $P_k$ is stable for $k\in \mathcal{K}$ and $P_k$ is possibly semi-stable for $k\in \mathcal{K}^\prime$. For $k \in \mathcal{K}$, the segmental phase $\Psi(P_k(\jw))$ can be treated as a special $\gamma$-segmental phase by identifying the phase center $\gamma^{\star}(P_k(\jw))$ to be a fixed inclined center. For simplicity, denote $\gamma_k(\omega)\coloneqq \gamma^{\star}(P_k(\jw))$ for $k \in \mathcal{K}$, and clearly $\gamma_k$ also belongs to the set $\mathcal{F}_k$ in \eqref{eq:frequency_gamma_set} due to Assumption~\ref{assum: singleton}. Without loss of generality, we only need to consider the existence of $\gamma_k\in \mathcal{F}_k$ for $k=1, 2, \ldots, m$ so that \eqref{eq:seg_small_phase} holds for all $\omega \in \interval{0}{\infty}$. 
 
 For $k=1, 2,\ldots, m$, the existence of $\gamma_k$ for \eqref{eq:seg_small_phase} implies that the existence of some $\hat{\gamma}_k:\mathbf{NC}\to\mathbb{R}$ for \eqref{eq:seg_small_phase}. Precisely, if $s=\jw$ and $\omega \in \interval{0}{\infty}\setminus \Omega$, set $\hat{\gamma}_k(s) =\gamma_k(\omega)$; if $s\in \semicir$, a continuous inclined center $\hat{\gamma}_k(s)$ always exists as the normalized numerical range $\mathcal{N}(P_k(s))$ changes continuously along $s\in \semicir$ in light of Assumptions~\ref{assum: system} and \ref{assum: singleton}. Having these understandings, we can repeat the same lines of reasoning as those stated in the proof of Theorem~\ref{thm:segmental_phase_stability} and only show the major differences. 
 
Firstly, note that we can analogously arrive at
 \be\label{eq:thm_gspt_nyquist}
 \textstyle \angle \lambda_i (P(s)) \in \sum_{k=1}^m \Psi_{\hat{\gamma}_k(s)}(P_k(s)) 
\subset \interval[open]{-\pi}{\pi}
\ee 
for all nonzero $\lambda_i(P(s))$, $s\in \semicir$ and $i=1, 2, \ldots, n$ on the basis of Theorem~\ref{thm:matrix_small_phase_cyc} and Assumptions~\ref{assum: system} and \ref{assum: singleton}, where $\omega_0\in \Omega^p$. In addition, for three cases of semi-stable $P_1(s)$ or $P_2(s)$, the constraints on the leading coefficient matrices can be obtained instead of \eqref{eq:thm_seg_residue_case1}, \eqref{eq:thm_seg_residue_case2} and \eqref{eq:thm_seg_residue_case3}, respectively:
\begin{enumerate} \renewcommand{\theenumi}{\textup{(\roman{enumi})}}\renewcommand{\labelenumi}{\theenumi}
\item 
$ \Psi_{\gamma_1(\omega_0)}(K_1)+ \textstyle \sum_{k=2}^{m}\Psi_{\gamma_k(\omega_0)}(P_k(\jw_0)) \subset \textstyle \interval{-\frac{\pi}{2}}{\frac{\pi}{2}}$;
\item 
$ \Psi_{\gamma_1(\omega_0)}(H_1)+ \textstyle \sum_{k=2}^{m}\Psi_{\gamma_k(\omega_0)}(P_k(\jw_0)) = 0$;
\item 
 $ \textstyle \sum_{k=1}^2 \Psi_{\gamma_k(\omega_0)}(K_k) +\sum_{k=3}^{m}\Psi_{\gamma_k(\omega_0)}(P_k(\jw_0))=0$.
\end{enumerate}
Combining \eqref{eq:seg_small_phase} and \eqref{eq:thm_gspt_nyquist} leads to that $\det \sbkt{I+P(s)} \neq 0$ for all $s\in \ccp \cup \{\infty\} \setminus \{ j\Omega_{\pm}^p\}$. Secondly, we can similarly show that any pole of $P(s)$ at $s=\jw_0$ with $\omega_0\in \Omega^p$ is not a pole of $\rbkt{I+P(s)}^{-1}$. \hspace*{\fill}~\QED 

 \section*{Acknowledgment}
 The authors are grateful to anonymous reviewers for many helpful comments and suggestions.
 
\section*{References}
\bibliographystyle{IEEEtran}
\bibliography{MIMOphase}

\begin{thebibliography}{10}
\providecommand{\url}[1]{#1}
\csname url@samestyle\endcsname
\providecommand{\newblock}{\relax}
\providecommand{\bibinfo}[2]{#2}
\providecommand{\BIBentrySTDinterwordspacing}{\spaceskip=0pt\relax}
\providecommand{\BIBentryALTinterwordstretchfactor}{4}
\providecommand{\BIBentryALTinterwordspacing}{\spaceskip=\fontdimen2\font plus
\BIBentryALTinterwordstretchfactor\fontdimen3\font minus \fontdimen4\font\relax}
\providecommand{\BIBforeignlanguage}[2]{{%
\expandafter\ifx\csname l@#1\endcsname\relax
\typeout{** WARNING: IEEEtran.bst: No hyphenation pattern has been}%
\typeout{** loaded for the language `#1'. Using the pattern for}%
\typeout{** the default language instead.}%
\else
\language=\csname l@#1\endcsname
\fi
#2}}
\providecommand{\BIBdecl}{\relax}
\BIBdecl

\bibitem{Tyson:78}
J.~J. Tyson and H.~G. Othmer, ``The dynamics of feedback control circuits in biochemical pathways,'' \emph{Prog. Theor. Biol.}, vol.~5, pp. 1--62, 1978.

\bibitem{Thron:91}
C.~Thron, ``The secant condition for instability in biochemical feedback control -- {\uppercase\expandafter{\romannumeral1}}. {T}he role of cooperativity and saturability,'' \emph{Bull. Math. Biol.}, vol.~53, no.~3, pp. 383--401, 1991.

\bibitem{Hori:11}
Y.~Hori, T.-H. Kim, and S.~Hara, ``Existence criteria of periodic oscillations in cyclic gene regulatory networks,'' \emph{Automatica}, vol.~47, no.~6, pp. 1203--1209, 2011.

\bibitem{Sontag:06}
E.~D. Sontag, ``Passivity gains and the ``secant condition'' for stability,'' \emph{Syst. Control Lett.}, vol.~55, no.~3, pp. 177--183, 2006.

\bibitem{Arcak:06}
M.~Arcak and E.~D. Sontag, ``Diagonal stability of a class of cyclic systems and its connection with the secant criterion,'' \emph{Automatica}, vol.~42, no.~9, pp. 1531--1537, 2006.

\bibitem{Scardovi:10}
L.~Scardovi, M.~Arcak, and E.~D. Sontag, ``Synchronization of interconnected systems with applications to biochemical networks: {A}n input-output approach,'' \emph{IEEE Trans. Autom. Control}, vol.~55, no.~6, pp. 1367--1379, 2010.

\bibitem{Hamadeh:11}
A.~Hamadeh, G.-B. Stan, R.~Sepulchre, and J.~Gon{\c{c}}alves, ``Global state synchronization in networks of cyclic feedback systems,'' \emph{IEEE Trans. Autom. Control}, vol.~57, no.~2, pp. 478--483, 2011.

\bibitem{Pates:23}
R.~Pates, ``A generalisation of the secant criterion,'' in \emph{Proc.~22nd IFAC World Congress}, Yokohama, Japan, 2023, pp. 9141--9146.

\bibitem{Chaffey:21j}
T.~Chaffey, F.~Forni, and R.~Sepulchre, ``Graphical nonlinear system analysis,'' \emph{IEEE Trans. Autom. Control}, vol.~68, no.~10, pp. 6067--6081, 2023.

\bibitem{Macfarlane:77}
A.~G. MacFarlane and I.~Postlethwaite, ``The generalized {N}yquist stability criterion and multivariable root loci,'' \emph{Int. J. Control}, vol.~25, no.~1, pp. 81--127, 1977.

\bibitem{Desoer:80}
C.~Desoer and Y.-T. Wang, ``On the generalized {N}yquist stability criterion,'' \emph{IEEE Trans. Autom. Control}, vol.~25, no.~2, pp. 187--196, 1980.

\bibitem{Smith:81}
M.~Smith, ``On the generalized {N}yquist stability criterion,'' \emph{Int. J. Control}, vol.~34, no.~5, pp. 885--920, 1981.

\bibitem{Zames:66}
G.~Zames, ``On the input-output stability of time-varying nonlinear feedback systems {P}art {\uppercase\expandafter{\romannumeral1}}: Conditions derived using concepts of loop gain, conicity, and positivity,'' \emph{IEEE Trans. Autom. Control}, vol.~11, no.~2, pp. 228--238, 1966.

\bibitem{Zhou:96}
K.~Zhou, J.~Doyle, and K.~Glover, \emph{Robust and Optimal Control}.\hskip 1em plus 0.5em minus 0.4em\relax Englewood Cliffs, NJ: Prentice Hall, 1996.

\bibitem{Jiang:94}
Z.~P. Jiang, A.~R. Teel, and L.~Praly, ``Small-gain theorem for {ISS} systems and applications,'' \emph{Math. Control. Signals, Syst.}, vol.~7, pp. 95--120, 1994.

\bibitem{Jiang:18}
Z.-P. Jiang and T.~Liu, ``Small-gain theory for stability and control of dynamical networks: A survey,'' \emph{Annu. Rev. Control.}, vol.~46, pp. 58--79, 2018.

\bibitem{Liu:11}
T.~Liu, D.~J. Hill, and Z.-P. Jiang, ``Lyapunov formulation of {ISS} cyclic-small-gain in continuous-time dynamical networks,'' \emph{Automatica}, vol.~47, no.~9, pp. 2088--2093, 2011.

\bibitem{Postlethwaite:81}
I.~Postlethwaite, J.~Edmunds, and A.~MacFarlane, ``Principal gains and principal phases in the analysis of linear multivariable feedback systems,'' \emph{IEEE Trans. Autom. Control}, vol.~26, no.~1, pp. 32--46, 1981.

\bibitem{Anderson:88}
B.~D.~O. Anderson and M.~Green, ``Hilbert transform and gain/phase error bounds for rational functions,'' \emph{IEEE Trans. Circuits Syst.}, vol.~35, no.~5, pp. 528--535, 1988.

\bibitem{Freudenberg:88}
J.~S. Freudenberg and D.~P. Looze, \emph{Frequency Domain Properties of Scalar and Multivariable Feedback Systems}.\hskip 1em plus 0.5em minus 0.4em\relax Berlin, Germany: Springer, 1988.

\bibitem{Chen:98}
J.~Chen, ``Multivariable gain-phase and sensitivity integral relations and design trade-offs,'' \emph{IEEE Trans. Autom. Control}, vol.~43, no.~3, pp. 373--385, 1998.

\bibitem{Owens:84}
D.~Owens, ``The numerical range: {A} tool for robust stability studies?'' \emph{Syst. Control Lett.}, vol.~5, no.~3, pp. 153--158, 1984.

\bibitem{Tits:99}
A.~L. Tits, V.~Balakrishnan, and L.~Lee, ``Robustness under bounded uncertainty with phase information,'' \emph{IEEE Trans. Autom. Control}, vol.~44, no.~1, pp. 50--65, 1999.

\bibitem{bar-on:90}
J.~R. Bar-on and E.~A. Jonckheere, ``Phase margins for multivariable control systems,'' \emph{Int. J. Control}, vol.~52, no.~2, pp. 485--498, 1990.

\bibitem{Anderson:73}
B.~D.~O. Anderson and S.~Vongpanitlerd, \emph{Network Analysis and Synthesis: {A} Modern Systems Theory Approach}.\hskip 1em plus 0.5em minus 0.4em\relax Englewood Cliffs, NJ: Prentice-Hall, 1973.

\bibitem{Lanzon:08}
A.~Lanzon and I.~R. Petersen, ``Stability robustness of a feedback interconnection of systems with negative imaginary frequency response,'' \emph{IEEE Trans. Autom. Control}, vol.~53, no.~4, pp. 1042--1046, 2008.

\bibitem{Petersen:10}
I.~R. Petersen and A.~Lanzon, ``Feedback control of negative-imaginary systems,'' \emph{IEEE Control Syst.}, vol.~30, no.~5, pp. 54--72, 2010.

\bibitem{Lanzon:22}
A.~Lanzon and P.~Bhowmick, ``Characterization of input--output negative imaginary systems in a dissipative framework,'' \emph{IEEE Trans. Autom. Control}, vol.~68, no.~2, pp. 959--974, 2023.

\bibitem{Zhao:22_NI}
D.~Zhao, C.~Chen, and S.~Z. Khong, ``A frequency-domain approach to nonlinear negative imaginary systems analysis,'' \emph{Automatica}, vol. 146, p. 110604, 2022.

\bibitem{Willems:72_2}
J.~C. Willems, ``Dissipative dynamical systems {P}art {\uppercase\expandafter{\romannumeral2}}: {L}inear systems with quadratic supply rates,'' \emph{Arch. Ration. Mech. Anal.}, vol.~45, no.~5, pp. 352--393, 1972.

\bibitem{Pates:19}
R.~Pates, C.~Bergeling, and A.~Rantzer, ``On the optimal control of relaxation systems,'' in \emph{Proc. 58th IEEE Conf. Decision and Control}, Nice, France, 2019, pp. 6068--6073.

\bibitem{Chaffey:23C}
T.~Chaffey, H.~J. van Waarde, and R.~Sepulchre, ``Relaxation systems and cyclic monotonicity,'' in \emph{Proc. 62nd IEEE Conf. Decision and Control}, Singapore, 2023, pp. 1673--1679.

\bibitem{Wang:20}
D.~Wang, W.~Chen, S.~Z. Khong, and L.~Qiu, ``On the phases of a complex matrix,'' \emph{Linear Algebra Appl.}, vol. 593, pp. 152--179, 2020.

\bibitem{Zhao_matrix:22}
D.~Zhao, A.~Ringh, L.~Qiu, and S.~Z. Khong, ``Low phase-rank approximation,'' \emph{Linear Algebra Appl.}, vol. 639, pp. 177--204, 2022.

\bibitem{Chen:21}
W.~Chen, D.~Wang, S.~Z. Khong, and L.~Qiu, ``A phase theory of multi-input multi-output linear time-invariant systems,'' \emph{SIAM J. Control Optim.}, vol.~62, no.~2, pp. 1235--1260, 2024.

\bibitem{Zhao:22}
\BIBentryALTinterwordspacing
D.~Zhao, W.~Chen, and L.~Qiu, ``When small gain meets small phase,'' \emph{arXiv}, 2022. [Online]. Available: \url{https://arxiv.org/abs/2201.06041}
\BIBentrySTDinterwordspacing

\bibitem{Liang:24}
J.~Liang, D.~Zhao, and L.~Qiu, ``Feedback stability under mixed gain and phase uncertainty,'' \emph{IEEE Trans. Autom. Control}, vol.~70, no.~2, pp. 1008--1023, 2025.

\bibitem{ChenJ:23}
J.~Chen, W.~Chen, C.~Chen, and L.~Qiu, ``Phase preservation of {$N$}-port networks under general connections,'' \emph{IEEE Trans. Autom. Control}, vol.~70, no.~4, pp. 2346--2361, 2025.

\bibitem{Chen:20j}
\BIBentryALTinterwordspacing
C.~Chen, D.~Zhao, W.~Chen, S.~Z. Khong, and L.~Qiu, ``Phase of nonlinear systems,'' \emph{submitted to IEEE Trans. Autom. Control, also in arXiv}, 2021. [Online]. Available: \url{https://arxiv.org/abs/2012.00692}
\BIBentrySTDinterwordspacing

\bibitem{Horn:59}
A.~Horn and R.~Steinberg, ``Eigenvalues of the unitary part of a matrix.'' \emph{Pacific J. Math.}, no.~4, pp. 541--550, 1959.

\bibitem{Furtado:01}
S.~Furtado and C.~R. Johnson, ``Spectral variation under congruence,'' \emph{Linear Multilinear A.}, vol.~49, no.~3, pp. 243--259, 2001.

\bibitem{Chen:22IFAC}
C.~Chen, W.~Chen, D.~Zhao, J.~Chen, and L.~Qiu, ``A cyclic small phase theorem for {MIMO} {LTI} systems,'' in \emph{Proc.~22nd IFAC World Congress}, Yokohama, Japan, 2023, pp. 1883--1888.

\bibitem{Green:12}
M.~Green and D.~J. Limebeer, \emph{Linear Robust Control}.\hskip 1em plus 0.5em minus 0.4em\relax Englewood Cliffs, NJ: Prentice Hall, 1995.

\bibitem{Auzinger:03}
W.~Auzinger, ``Sectorial operators and normalized numerical range,'' \emph{Appl. Numer. Math.}, vol.~45, no.~4, pp. 367--388, 2003.

\bibitem{Lins:18}
B.~Lins, I.~M. Spitkovsky, and S.~Zhong, ``The normalized numerical range and the {D}avis--{W}ielandt shell,'' \emph{Linear Algebra Appl.}, vol. 546, pp. 187--209, 2018.

\bibitem{Wielandt:67}
H.~Wielandt, \emph{Topics in the Analytic Theory of Matrices}.\hskip 1em plus 0.5em minus 0.4em\relax Madison, WI: University of Wisconsin Lecture Notes, 1967.

\bibitem{Desoer:75}
C.~A. Desoer and M.~Vidyasagar, \emph{Feedback Systems: {I}nput-Output Properties}.\hskip 1em plus 0.5em minus 0.4em\relax New York, NY: Academic Press, 1975.

\bibitem{Griggs:07}
W.~M. Griggs, B.~D. Anderson, and A.~Lanzon, ``A ``mixed'' small gain and passivity theorem in the frequency domain,'' \emph{Syst. Control Lett.}, vol.~56, no. 9-10, pp. 596--602, 2007.

\bibitem{Forbes:10}
J.~R. Forbes and C.~J. Damaren, ``Hybrid passivity and finite gain stability theorem: {S}tability and control of systems possessing passivity violations,'' \emph{IET Control. Theory Appl.}, vol.~4, no.~9, pp. 1795--1806, 2010.

\bibitem{Patra:11}
S.~Patra and A.~Lanzon, ``Stability analysis of interconnected systems with ``mixed'' negative-imaginary and small-gain properties,'' \emph{IEEE Trans. Autom. Control}, vol.~56, no.~6, pp. 1395--1400, 2011.

\bibitem{Liu:15}
K.-Z. Liu, ``A high-performance robust control method based on the gain and phase information of uncertainty,'' \emph{Int. J. Robust Nonlinear Control}, vol.~25, no.~7, pp. 1019--1036, 2015.

\bibitem{Lestas:12}
I.~Lestas, ``Large scale heterogeneous networks, the {D}avis--{W}ielandt shell, and graph separation,'' \emph{SIAM J. Control Optim.}, vol.~50, no.~4, pp. 1753--1774, 2012.

\bibitem{Chaffey:22_rolled}
T.~Chaffey, ``A rolled-off passivity theorem,'' \emph{Syst. Control Lett.}, vol. 162, p. 105198, 2022.

\bibitem{Woolcock:23}
L.~Woolcock and R.~Schmid, ``Mixed gain/phase robustness criterion for structured perturbations with an application to power system stability,'' \emph{IEEE Control Syst. Lett.}, pp. 3193 -- 3198, 2023.

\bibitem{Marelli:23}
D.~E. Marelli and M.~Fu, ``Stability of networked nonlinear systems: {G}eneralization of small-gain theorem and distributed testing,'' \emph{Automatica}, vol. 152, p. 110937, 2023.

\bibitem{ChenKW:24}
K.~Chen and A.~Astolfi, ``Active nodes of network systems with sum-type dissipation inequalities,'' \emph{IEEE Trans. Autom. Control}, vol.~69, no.~6, pp. 3896--3911, 2024.

\end{thebibliography}

\end{document}